\documentclass[9.5pt]{IEEEtran}

\usepackage[colorlinks,urlcolor=blue,linkcolor=blue,citecolor=blue]{hyperref}

\usepackage{enumitem}

\usepackage{array}
\usepackage[hang]{footmisc}

\usepackage[subrefformat=parens,labelformat=parens,caption=false,font=footnotesize]{subfig}
\usepackage{amsmath}
\usepackage[nolist,nohyperlinks]{acronym}
\usepackage{amsmath}

\begin{acronym}
    \acro{4G}{fourth generation}
    \acro{5G}{fifth generation}
    \acro{AoA}{angle of arrival}
   	\acro{AoD}{angle of departure}
   	\acro{AP}{access point}
    \acro{BCRLB}{Bayesian CRLB}
    \acro{BG}{beam group}
    \acro{BS}{base stations}
    \acro{BPP}{binomial point process}
    \acro{BP}{broadcast probing}
	\acro{CDF}{cumulative density function}
    \acro{CCDF}{complementary cumulative density function}
    \acro{CF}{closed-form}
    \acro{CKM}{constrained k-means}
    \acro{DAS}{distributed antenna system}
    \acro{CRLB}{Cramer-Rao lower bound}
    \acro{ECDF}{empirical cumulative distribution function}
    \acro{EI}{Exponential Integral}
    \acro{eMBB}{enhanced mobile broadband}
    \acro{FIM}{Fisher Information Matrix}
    \acro{GoF}{goodness-of-fit}
    \acro{GPS}{global positioning system}
    \acro{GE}{group exploration}
    \acro{GNSS}{global navigation satellite system}
    \acro{HetNet}{heterogeneous network}
    \acro{IoT}{internet of things}
    \acro{IIoT}{industrial internet of things}
    \acro{ILP}{integer linear program}
    \acro{KM}{k-means}
    \acro{KPI}{key performance indicator}
    \acro{KHM}{k-harmonic means}
    \acro{WKHM}{weighted k-harmonic means}
    \acro{KC}{k-centers}
    \acro{LOS}{line of sight}
    \acro{LLR}{log-likelihood ratio}
    \acro{MAB}{multi-armed bandit}
    \acro{MBS}{macro base station}
    \acro{MD}{\ac{MD}}
    \acro{MEC}{mobile-edge computing}
    \acro{mIoT}{massive internet of things}
    \acro{MIMO}{multiple input multiple output}
    \acro{mm-wave}{millimeter wave}
    \acro{mMTC}{massive machine-type communications}
    \acro{MS}{mobile station}
    \acro{MVUE}{minimum-variance unbiased estimator}
    \acro{NLOS}{non line-of-sight}
    \acro{OFDM}{orthogonal frequency division multiplexing}
    \acro{PAC}{probably approximately correct}
    \acro{PDF}{probability density function}
    \acro{PGF}{probability generating functional}
    \acro{PLCP}{Poisson line Cox process}
    \acro{PLT}{Poisson line tessellation}
    \acro{PLP}{Poisson line process}
    \acro{PPP}{Poisson point process}
    \acro{PV}{Poisson-Voronoi}
    \acro{QoS}{quality of service}
    \acro{RAT}{radio access technique}
    \acro{RIC}{radio access network intelligent controller}
    \acro{RL}{reinforcement-learning}
    \acro{RN}{radio node}
    \acro{RSSI}{received signal-strength indicator}
    \acro{RSRP}{reference signal received power}
    \acro{BS}{base station}
    \acro{SINR}{signal to interference plus noise ratio}
    \acro{SG}{stochastic geometry}
    \acro{SNR}{signal to noise ratio}
    \acro{SSB}{synchronization signal block}
    \acro{SWIPT}{simultaneous wireless information and power transfer}
    \acro{TS}{Thompson Sampling}
    \acro{TS-CD}{TS with change-detection}
    \acro{KS}{Kolmogorov-Smirnov}
    \acro{UCB}{upper confidence bound}
	\acro{ULA}{uniform linear array}
    \acro{UPA}{uniform planar array}
	\acro{UE}{user equipment}
 	\acro{URLLC}{ultra-reliable low-latency communications}
    \acro{V2V}{vehicle-to-vehicle}    
    \acro{wpt}{wireless power transfer}
\end{acronym}

\usepackage{graphicx,wrapfig,lipsum}
\usepackage{rotating}
\usepackage{amsmath, amssymb, amsthm}
\usepackage{stmaryrd}
\usepackage{tikz}
\usepackage{verbatim}
\usepackage{url}
\usepackage[utf8]{inputenc}
\usepackage[english]{babel}
\newtheorem{theorem}{Theorem}[]
\newtheorem{corollary}{Corollary}[]

\newtheorem{lemma}[]{Lemma}

\usepackage{multicol}
\newtheorem{rmk}{Remark}[]

\usepackage{algorithm}
\usepackage{algpseudocode}
\usepackage{scalerel}
\usepackage{multicol}
\usepackage{setspace}

\usepackage[noadjust]{cite}
\usepackage{booktabs}
\usepackage{derivative}
\usepackage{bbm}
\usepackage[normalem]{ulem}
\definecolor{Ao}{rgb}{0.0, 0.62, 0.1}
\usepackage{color}
\usepackage{dsfont}
\usepackage{bm}
\usepackage{setspace}
\usepackage{cellspace}
\usepackage{array}
\usetikzlibrary{automata}
\usetikzlibrary{arrows}
\usetikzlibrary{shapes.geometric, arrows}
\usepackage[]{footmisc}
\usetikzlibrary{positioning}
\usetikzlibrary{calc}
\usetikzlibrary{arrows}
\hypersetup{nolinks=true}
\allowdisplaybreaks

\usepackage{caption}
\usepackage{textcomp}

\usepackage{mathtools}
\begin{document}
	\title{\fontsize{22.8}{27.6}\selectfont A Framework for Enterprise Network Dimensioning}
   \author{Gourab Ghatak
  \thanks{G. Ghatak is with the Department of Electrical Engineering, IIT Delhi, India 110016;
  \vspace{-0.5cm}}}
\date{}

\maketitle
\begin{abstract}
We study radio node (RN) placement for indoor enterprise networks. Using stochastic geometry (SG), we derive the meta-distribution (MD) of the SINR for a test user equipment (UE), with and without cooperation from outdoor macro base stations (MBSs), and compare these results with an integer linear programming (ILP) approach. SG provides an estimate of the required number of RNs but not their locations, while ILP can yield inaccurate local optima and requires high computational power. To address this, we investigate clustering-based algorithms for initializing RN locations using UE location distributions. Along with standard methods, we propose a weighted $k$-harmonic means (WKHM) clustering strategy tailored to maximize SINR. We then introduce a constrained sequential minimum cut algorithm, \texttt{SeqMinCut}, to merge multiple RNs into larger cells and further improve SINR. This is the first work that integrates SG-based statistical analysis, optimization, and clustering to obtain system design insights, dimensioning rules, and planning strategies for enterprise 5G. 
\end{abstract}

\section{Introduction}
The advent of \ac{5G} is transforming many industries by offering unprecedented connectivity, low latency, and high data rates and reliability~\cite{andrews2014will,rappaport2013mmwave}. Enterprise 5G networks, in particular, can boost operational efficiency and productivity by tailoring the network to specific use cases and their requirements~\cite{ordonez2019use,foukas2017slicing}. As businesses pursue digitization and automation to stay competitive and meet evolving consumer expectations, robust and customized network infrastructure becomes essential. Successful deployment and optimization of enterprise 5G, however, require careful planning and deep system design insights, including mapping application and use-case requirements to network requirements and \ac{QoS} metrics such as rate, reliability, and coverage.

In this context, this paper presents a framework for planning and deploying \acp{RN} in indoor scenarios to enable enterprise 5G. Through a comprehensive analysis, we address three key questions for indoor 5G planning: {\it i) how many \acp{RN} to deploy?, ii) where to deploy them?, and iii) how to enable clustered cell transmission using \ac{DAS} for enhanced \ac{UE} experience.} In doing so, we identify the main considerations and best practices for successful implementations, and clarify the evolving role of enterprise 5G connectivity in modern business environments.
\section{Background and Related Work}

\begin{table*}[t]
\small{
\centering
\caption{Comparison with Prior Work}
\label{tab:comparison}
\begin{tabular}{|p{3cm}|p{3cm}|p{3cm}|p{3cm}|p{3.6cm}|}
\hline
\textbf{Feature / Criteria} & \textbf{Stochastic Geometry (SG) Approaches} \cite{elsawy2016modeling,hmamouche2021new,ghatak2018coverage,haenggi2021meta,haenggi2021meta2} & \textbf{Optimization-Based Approaches} \cite{krishnamachari2000optimization,ji2002methods,lee2000cell,molina1999automatic,ahmadi2024wireless,abrishambaf2019study} & \textbf{Commercial Tools (e.g., iBwave)} \cite{kunisch2006locally,ahmadi2024wireless,haron2021performance,zhang2021radio, domainremcom} & \textbf{This Work} \\
\hline
Primary Methodology & SG-based probabilistic modeling & ILP, GA, Tabu, heuristic search & Ray-tracing + empirical models & Joint SG, ILP, clustering, and graph partitioning \\
\hline
Performance Metric & SINR/RSRP meta-distribution & Coverage, signal strength, UE distribution & RSRP heuristics & MD, RSRP, SINR across all realizations \\
\hline
Insight Provided & RN count distribution & Scenario-specific optimal count & Heuristic coverage regions & Transferable RN dimensioning rules and optimization flow \\
\hline
Location Insight & None & Scenario-specific local optima & Geometry-based suggestions & UE-aware RN location optimization via WKHM \\
\hline
User Distribution Awareness & Uniform distribution assumed & Some methods include UE density & Limited or manual & Fully incorporated using KM, CKM, WKHM \\
\hline
Interference Modeling & Inter-tier interference via SG & Often ignored or simplified & Ray-tracing simulation & Modeled in SG and ILP jointly \\
\hline
Computational Complexity & Low (analytical) & High (combinatorial) & High (manual/simulation) & Balanced: SG estimates + ILP refinement \\
\hline
Clustered Cell Design & Not addressed & Not addressed or manually done & High redundancy setting & SeqMinCut-based RN clustering \\
\hline
Support for DAS & No & No & Partial & Yes; includes admission control \\
\hline
Novel Algorithms & – & – & – & WKHM (weighted KHM), SeqMinCut clustering \\
\hline
Adaptability to Scenario Changes & Moderate (model-dependent) & Poor (new optimization) & Poor (manual tuning) & High (statistical + optimization + clustering) \\
\hline
Tool Availability & N/A & N/A & Commercial license & Open MATLAB Toolbox \cite{codes} \\
\hline
\end{tabular}
}
\end{table*}

\subsection{Statistical analyses with \ac{SG}} Over the last couple of decades, \ac{SG} has proven to be a useful modeling and analysis tool for wireless networks~\cite{elsawy2016modeling, hmamouche2021new, ghatak2018coverage, ghatak2019small, ghatak2020beamwidth, ghatak2021stochastic, shah2024binomial}. { In particular, it has been extensively employed to analyze \acp{HetNet}, where multiple tiers of base stations are typically modeled as independent PPPs~\cite{andrews2023introduction,haenggi27stochastic,dhillon2012modeling, ghosh2012heterogeneous, lu2015stochastic, di2015stochastic}. These studies have provided fundamental insights into coverage, rate distributions, and multi-tier interference management. Our proposed architecture, consisting of a PPP-distributed macrocell tier and an indoor small-cell tier, is structurally analogous to such K-tier HetNet models. However, the key distinction in this work lies in modeling the indoor \acp{RN} using a \ac{BPP}, which captures the finite and fixed number of RNs typically present in enterprise deployments.} Albeit in the first instance this appears as a minor modification of the baseline $K-$tier \ac{HetNet} model, { due to the BPP formulation leads to non-stationarity and the absence of a ``typical UE'', thereby necessitating a test-UE based analysis. This introduces analytical challenges that are not encountered in standard HetNet frameworks and requires new tools for deriving the distance distribution and the SINR meta-distribution. In particular, now we need the conditional distributions, given the test-user location which can be remarkably different for a UE present at the center of the room in contrast to a UE at the room edge. This directly captures the varying nature of UE experience conditioned on the location that is missed by classical PPP models.} In \ac{SG}, the locations of network elements such as \acp{RN}, blockages, and \acp{UE} are modeled as spatial stochastic processes. These probabilistic models let planners statistically evaluate key performance metrics such as coverage probability, interference, and spectral efficiency~\cite{ghatak2018coverage, ganti2009stochastic}. In particular, the \ac{MD} framework of \ac{SG}~\cite{haenggi2021meta,haenggi2021meta2} yields fine-grained per-link performance insights for \acp{UE}, enabling network deployment planning for target \acp{KPI} such as coverage, delay, and performance variance. \ac{MD} has been derived for multi-tier networks~\cite{wang2018sir} and with \ac{RN} cooperation~\cite{cui2017sir}, mainly using Poisson models for node locations. However, for indoor networks, the Poisson assumption for \ac{RN} locations may not hold. Afshang~{et al.}~\cite{afshang2017fundamentals} present fundamentals of modeling finite networks using the \ac{BPP}. A multi-tier analysis that includes a fixed number of indoor \acp{RN} and a Poisson-distributed outdoor \ac{MBS} tier has not yet been reported. This gap is critical for the present work, as most enterprise networks will opportunistically use both indoor \acp{RN} and outdoor \acp{MBS} to optimize performance.

Despite its tractability, \ac{SG} analysis has three main drawbacks. First, it relies on simplifying assumptions, such as uniform spatial distributions and homogeneous propagation environments. Second, \ac{SG} models struggle to capture dynamic factors, including UE mobility and blockages. Third, they only assess the impact of RN density or number, offering no guidance on RN placement. Thus, while \ac{SG} is useful for 5G network planning, its applicability and accuracy can be limited in some deployments, motivating further development of radio network planning tools, which we pursue in this work.

\subsection{Optimization Approaches} The \ac{RN} placement problem has been widely studied from an optimization perspective, mainly for outdoor networks; see~\cite{krishnamachari2000optimization, ji2002methods, lee2000cell,hurley2002planning, laiho2002radio, binzer2000radio, tun2021radio, corre2009three, gamst2009cellular, ayoubi2025advanced, su20225g, zhang2023rme}. Related facility-location formulations and approximation algorithms (e.g., $k$-median / $p$-median style models) are also widely used as abstractions for access-point placement~\cite{daskin2013network}. For indoor \ac{RN} placement, commonly used algorithms include genetic algorithms~\cite{molina1999automatic}, Tabu search~\cite{wu2007optimization}, binary {ILP}~\cite{wong2006base}, greedy algorithms~\cite{kim2007multiple}, and Ngadiman’s algorithm~\cite{ngadiman2005new}. If computational complexity is not a concern and the search space is discrete, exhaustive search can yield the optimal \ac{RN} locations. Pujji~\cite{pujji2012optimisation} provides a comprehensive analysis of these methods and proposes a hybrid algorithm that balances computational complexity and placement optimality. 

\ac{UE} locations can guide \ac{RN} placement. While $k$-means clustering suits static UEs~\cite{abrishambaf2019study}, mobile UEs require more robust approaches, such as the two-stage optimization of~\cite{chatterjee2021joint}, later applied in~\cite{topal2023optimal,abdel2024robust}. Even statistical UE information can aid placement when exact locations are unknown. However, such optimization-based methods are scenario-specific and offer limited system-level insight, motivating the joint \ac{SG}–optimization framework proposed in this work.

\subsection{Commercial Software for Network Planning}
Indoor radio network planning has been extensively studied to meet the growing demand for seamless wireless connectivity in complex indoor environments~\cite{louro20203d, albanese2022ris, xia2023digital}.
Early planning relied on empirical models like the COST 231 multi-wall model~\cite{action1999231} and the ITU-R indoor model~\cite{series2015propagation}, which estimate path loss from statistically fitted measurements. Although computationally efficient, these models lack the detail needed to capture complex wave–structure interactions, hindering accurate characterization of performance variation across the area. To overcome this, deterministic ray-tracing models gained prominence. Works such as Kunisch and Pamp~\cite{kunisch2006locally} introduced advanced ray-tracing to accurately model reflections, diffractions, and transmissions. These capabilities are integrated into commercial tools like iBwave Design~\cite{ahmadi2024wireless}, Remcom~\cite{domainremcom}, WinProp~\cite{haron2021performance}, and Ranplan Professional~\cite{zhang2021radio}. Their core strength is accurate ray-tracing predictions, followed by manual \ac{RN} placement. Some tools have recently added automatic radio node placement, but their algorithms use only the geometry and dimensions of the deployment area and do not perform rigorous optimization. 

\subsection{Key Contributions, Novelty, and Insights}
We provide a multi-step joint \ac{RN} number, placement, and association optimization framework for indoor enterprise \ac{5G} networks. The salient contributions of this work are discussed below. Corresponding to each contribution, we highlight the key novelty, that, to the best of our knowledge, has not been addressed in literature.
\begin{itemize}[leftmargin=*]
    \item First, by approximating the deployment area as a disk and using tools from \ac{SG}, we derive the \ac{MD} of a uniformly located \ac{UE}. This enables the estimation of the number of required \acp{RN} for not only satisfying a given downlink \ac{SINR} requirement but also a given per-UE throughput requirement. We investigate two cases: i) prioritized \ac{RN} association and ii) \ac{MBS} cooperation, and study the consequent decrease in the required number of \acp{RN} with cooperation. It is important to highlight that although {\it dominance over macro} is a key feature for indoor network designers, it has not been studied in literature.

    {\bf Key Novelty:} (1) Characterization of the \ac{MD} of SINR for a two-tier network consisting of a \ac{BPP} tier and an \ac{PPP} tier, and (2) Location dependent \ac{MD} analysis. This is challenging due to non existence of a {\it typical UE} owing to the asymmetry in the point processes from the \ac{UE} perspective.
    \item To compare the insights derived from the \ac{MD}, we formulate an \ac{ILP} to jointly optimize the number and the locations of the \acp{RN}. The \ac{MD} analysis provides a reliable upper bound to the \ac{ILP} solution.

    {\bf Key Novelty:} (1) In contrast to~\cite{wong2006base}, we include the impact of \ac{MBS} interference in deriving the optimal number of RNs.
    \item The \ac{ILP} formulation provides only a local optima for the placement of the \acp{RN}. Furthermore, it does not take into account the \ac{UE} densities in the deployment area. To alleviate this, based on the derived SG based upper bound on the number of \acp{RN}, we investigate the initial \ac{RN} placement methods leveraging off-the-shelf clustering algorithms such as \ac{KM}, \ac{CKM}, \ac{KHM}, and \ac{KC}. We highlight the advantages and drawbacks of each of these strategies. 

    {\bf Key Novelty:} (1) In order to strike a balance between the downlink \ac{RSRP} and association, we propose a novel clustering method called \ac{WKHM}. We derive its centroid-update step and demonstrate that for the association-unconstrained case, it outperforms all its competitors in terms of donwlink \ac{RSRP}; (2) This is the first work that combines a \ac{SG} analysis with a placement optimization to derive network dimensioning rules, thereby attempting to alleviate the limitations of a \ac{SG} only analyses that are limited to determining the number of RNs. 
    \item We propose a novel algorithm to combine the RNs to form clustered cells. Forming clustered cells by combining multiple RNs forms the basis of \acp{DAS}, subject to a maximum \ac{UE} association per cell have the advantage of improving the \ac{SINR} and \ac{RSRP} of the \acp{UE}. The proposed algorithm called \texttt{SeqMinCut} employs a vertex weight constrained sequential Stoer-Wagner mincut of a graph starting from a fully connected graph to minimize the inter-RN interference.
    
    {\bf Key Novelty:} In literature, only geographical proximity and backhaul capacity determine the \ac{RN} cooperation, without any consideration to the number of associated \acp{UE}. Thus, our proposal implements a preemptive admission control in order to maintain the pilot orthogonality of intra-cell \acp{UE}~\cite{bjornson2017massive}.

    \item We leverage the above results to provide a unified framework for complete planning of indoor enterprise networks. The flow of planning so obtained is organized as a indoor network planning MATLAB toolbox consisting of functions to implement each of the above steps~\cite{codes}.
\end{itemize}

{
\subsection{Dominant contribution}
The central contribution of this paper is a reliability-to-deployment framework for bounded enterprise networks. The framework is organized as a sequence of explicit stage-wise mappings
The stages are not independent solutions to unrelated problems. The finite-area SG/meta-distribution analysis provides an RN cardinality but no location information. The candidate-site ILP converts this cardinality into a geometry- and deployability-aware placement. WKHM then incorporates the spatial demand distribution and received-signal fractions, while SeqMinCut converts the placed RNs into operational cells under admission constraints. The final commercial-tool comparison evaluates the complete deployment produced by these handoffs. Thus, the principal advance is the explicit and technically consistent coupling of reliability analysis, constrained placement, and cell formation into one implementable planning procedure.

The individual technical contributions support this central mapping: the bounded mixed BPP-PPP analysis supplies the reliability-to-cardinality step; the deployability-aware ILP supplies the cardinality-to-placement step; WKHM supplies signal-aware demand adaptation; and SeqMinCut supplies load-constrained cell formation. This organization also makes the limitations transparent: the SG output is a conservative baseline rather than a location optimum, and detailed geometry and installation constraints enter only in the subsequent placement and validation stages.
}

{
The proposed framework should be interpreted as a reliability-to-deployment pipeline. The finite-area stochastic geometry analysis first converts a reliability specification into a conservative RN-count recommendation. The candidate-site optimization then converts this cardinality recommendation into a deployable placement subject to radio and installation constraints. Finally, WKHM and SeqMinCut refine the user-aware placement and cell structure, respectively, and the resulting deployment is validated under detailed indoor propagation. In addition to the baseline results, we evaluate the sensitivity of the RN count to footprint shape, UE non-uniformity, and propagation uncertainty; provide a signal-fraction interpretation and ablation of the WKHM weighting; report confidence intervals and paired statistical tests; and distinguish the enterprise-specific inputs from the transferable analytical and optimization components.
}

{
\subsection{Scope and generality}
The framework separates scenario-specific inputs from transferable methodology. The enterprise-specific inputs are the bounded indoor floorplan, material-dependent propagation and blockage model, finite candidate installation set, cabling and power availability, prohibited or inaccessible zones, maintenance requirements, and platform-specific admission limits. These quantities define a particular planning instance, but they are not intrinsic restrictions of the analytical or optimization methods.

The transferable components are: (i) location-conditioned finite-area reliability analysis, (ii) conversion of a percentile-reliability target into an RN-count baseline, (iii) placement over an arbitrary designer-provided candidate set with linear feasibility and cost constraints, (iv) signal-aware refinement using a spatial demand distribution, and (v) graph partitioning under load or admission constraints. A campus deployment can instantiate the same pipeline with a multi-building indoor-outdoor propagation model; an industrial facility can introduce safety-clearance, machinery-obstruction, and deterministic installation constraints; and a public venue can replace the static UE law by crowd-density maps for different event states. The stage-wise interfaces of the pipeline remain unchanged.

The present framework is a long-term planning method. If user mobility, traffic bursts, or blockage dynamics evolve on the same time scale as association and scheduling, an additional temporal traffic/queueing and reconfiguration layer is required. Such short-time-scale control is complementary to, rather than a replacement for, the placement and dimensioning layer studied here.
}



\section{System Model and Objectives}
First, we assume that the deployment scenario is characterized by its area. This is due the fact that individual deployment scenarios can be of varied shapes and complex propagation environments which inhibit a statistical analysis. { In order to make the \ac{SG} analysis tractable, we adopt a simplified system model where the deployment area is approximated as a disk and the pathloss depends only on the Euclidean distance with a fixed exponent. While this abstraction facilitates analytical characterization of the SINR meta-distribution, it does not capture the full range of indoor propagation phenomena such as wall penetration losses, shadowing, or attenuation due to furniture and obstacles. We emphasize that this choice is not intended to replace realistic modeling, but rather to provide statistical insights and general dimensioning rules. To address practical complexities, we later incorporate wall blockages and irregular geometries in the ILP-based framework (Section~IV.C). This two-step approach allows us to balance analytical tractability with practical relevance.} The \acp{RN} are assumed to be installed on the ceiling, with height $g$. Consider the number of \acp{RN} be $N_{\rm r}$. Furthermore, let the number of \acp{UE} be $N_{\rm u}$. First, we assume that they are uniformly distributed in the floor. Thus, the UE density is $\lambda_{\rm u} = \frac{N_{\rm u}}{\pi R^2}$ UEs/m$^2$ if the deployment area is a disk of radius $R$, while it is $\lambda_{\rm u} = \frac{N_{\rm u}}{R^2}$ if it is a square of side $R$.

Let the per-UE request arrival rate be $\lambda_{\rm a}$ bits/s/UE, giving a traffic density $\sigma = \lambda_{\rm u} \lambda_{\rm a}$ bits/s/m$^2$. The carrier frequency is $f_{\rm c}$ and the aggregate bandwidth is $B$. Let $\bar{\xi} = \{\xi_1, \xi_2, \ldots, \xi_{N_u}\}$ denote the downlink \ac{SINR} of the \acp{UE}. The \acp{RN} are located at $\Phi_{\rm r} = \{{\bf m}_1, {\bf m}_2, \ldots, {\bf m}_{N_{\rm r}}\}$, partitioned into $N_{\rm c} \leq N_{\rm r}$ clusters (cells), each with one or more \acp{RN}. Each \ac{UE} associates with one cell, so the \acp{UE} are partitioned into $N_{\rm c}$ groups with locations $\Phi_{\rm u1}, \Phi_{\rm u2}, \ldots, \Phi_{uN_{\rm c}}$, each served by its cell. Downlink transmissions experience inter-cluster interference.

Each \ac{RN} transmits with power $P$. A transmit–receive pair at distance $d$ has received power $KhPd^{-\alpha}$, where $K = \left(\frac{c}{4\pi f_{\rm c}}\right)^2$ is the path-loss constant, $\alpha$ is the path-loss exponent, and $h$ is the small-scale fading. We denote by $\ell_{\rm r}(d)$ and $\ell_{\rm m}(d)$ the path-loss to the receiver from indoor \acp{RN} and from \acp{MBS}, respectively. According to 3GPP, choosing an appropriate $\alpha$ (e.g., 3) without explicit wall-penetration modeling can emulate indoor propagation~\cite{3GPPmodel,3gpp38901,winner22007}. Let the channel noise power be $N_0$.

The objective is to minimize the number of \acp{RN} while meeting downlink coverage and capacity requirements. The optimization variables are the number of \acp{RN}, the number of clusters, the clustering assignments, and the \ac{RN} locations. Mathematically, the problem is to minimize $|\Phi_{\rm r}|$ subject to a downlink performance constraint $f(\bar{\xi}) \geq f_0$, e.g., minimum or mean \ac{RSRP}, or mean \ac{SINR}. We also investigate individual cell load, which depends on the number and locations of \acp{RN} and the \acp{UE} traffic demands~\cite{ghatak2018accurate}, making the problem challenging because coverage and capacity requirements often conflict. To simplify, we only study how our algorithms affect load and do not include it explicitly in the optimization. The number of \acp{UE} served by a cell is further limited by admission control.

\section{\ac{MD} Analysis and \ac{ILP}}
\label{sec:MD}
The \ac{MD} of SINR~\cite{haenggi2016meta} provides a finer characterization by quantifying the fraction of users that achieve a target reliability level. { In other words, while the coverage probability answers the question \emph{``what is the probability of success for a randomly chosen user?''}, the MD answers the stronger question \emph{``what fraction of users achieve a given level of reliability?''}. This distinction is crucial for enterprise networks, where performance guarantees must hold not only on average but also for individual links. In particular, in our setting, the MD represents that the faction of \acp{UE} over all network realizations that achieve and SNR/RSRP larger than a threshold is at least some reliability threshold. Therefore, the MD serves as a central performance metric in our framework, enabling the derivation of dimensioning rules that are both statistically rigorous and practically relevant.}

We first consider a network where each cell contains a single \ac{RN}, i.e., no \ac{RN} clustering. A \ac{UE} served by an \ac{RN} experiences interference from all other \acp{RN}. With no prior information on \ac{UE} locations, they are assumed to be uniformly distributed. For the \ac{MD} analysis, the deployment region is a disk of radius $R$, in which $N_{\rm r}$ \acp{RN} form a \ac{BPP} $\Phi_{\rm r} \subset \mathcal{B}((0,0),R)$. The \acp{MBS}, with locations $\Phi_{\rm m}$, are modeled as a \ac{PPP} in the annular region $\mathcal{B}((0,0),R_{rm m}) \setminus \mathcal{B}((0,0),R)$. Owing to the bounded deployment regions, both $\Phi_{\rm r}$ and $\Phi_{\rm m}$ are translation-variant, precluding the definition of a network-wide \emph{typical} \ac{UE} as in standard \ac{SG} analysis. {In other words, the global stationarity assumption underlying the Palm distribution, and thereby the concept of a typical point, is not valid here. To address this, we adopt a \emph{test UE} located at a deterministic distance $r_0$ from the origin, and condition 
all performance metrics on this location. This construction is consistent with Palm theory. For a point process  $\Phi$, the reduced Palm distribution $P^!_x$ satisfies
\[
\mathbb{E}[f(\Phi) \mid \text{point at } x] = \int f(\Phi \cup \{x\}) \, dP^!_x(\Phi).
\]
In stationary PPP models, $P^!_x$ is translation invariant, giving rise to the typical UE. In our bounded BPP + PPP 
setting, we instead define the test UE and condition on its distance $r_0$, thereby obtaining a mathematically rigorous 
alternative to the typical UE assumption. All subsequent SINR meta-distribution derivations are therefore carried out 
with respect to this test UE framework.} Leveraging the isotropic nature of the scenario, let us consider our test UE to be at $(r_0, 0)$.

\begin{rmk}
{Since the RN process is a bounded BPP, stationarity is lost and 
a typical UE cannot be defined in the classical sense. Therefore, when introducing Lemma~2, we define a \emph{test UE} at distance $r_0$ from the origin and condition the analysis on this location. This corresponds to a reduced Palm distribution centered at $r_0$ rather than translation-invariant Palm conditioning. Consequently, all results derived henceforth are with respect to this test UE framework, which rigorously replaces the typical UE notion. }
\end{rmk}

\subsection{RN Prioritized Connection}
The instantaneous \ac{SINR} at the test \ac{UE} is
\begin{align}
    \xi = \frac{KP_{\rm r}\ell_{\rm r}(d)h_1}{N_0 + \sum_{i \in \Phi_{\rm r} \backslash {{\bf x}_1}} P_{\rm r} \ell_{\rm s}(d_i)h_i  + \sum_{j \in \Phi_{\rm m}} P_{\rm m} \ell_{\rm m}(d_i)h_j},
\end{align}
where $h_1$ is the channel fading power from the serving \ac{RN}. On the other hand, $h_i$ and $h_j$ are the channel fading powers from the interfering \acp{RN} and \acp{MBS} respectively. We consider a transmission attempt to be successful if the instantaneous rate exceeds the instantaneous traffic arrival rate. Thus, leveraging Shannon's formula, the success probability, conditioned on the \ac{MBS} and \ac{RN} locations is defined as
\begin{align}
    P_{\rm s} &= \mathbb{P}\left(\frac{B}{N_{\rm a}} \log\left(1 + \xi\right) \geq \lambda_{\rm a} | \Phi_{\rm r}, \Phi_{\rm m}\right), \nonumber \\
    &= \mathbb{P}\left(\xi \geq \gamma | \Phi_{\rm r}, \Phi_{\rm m}\right), \nonumber
\end{align}
where $\gamma = 2^{\frac{\lambda_{\rm a} N_{\rm a} A}{B\pi R^2}} - 1$. $N_{\rm a}$ is the number of associated \acp{UE} to the serving \ac{RN} of the test \ac{UE}.

\begin{lemma}
    The conditional success probability of the test UE at $r_0$, given a realization of $\Phi_{\rm r}$ and $\Phi_{\rm m}$ is
    \begin{align}
    P_{\rm s}(r_0) = &\exp\left(\frac{-\gamma N_0}{KP_{\rm r}\ell_{\rm r}(d)}\right) \cdot\nonumber  \\
        &\left[\prod_{\Phi'_{\rm s}}\frac{1}{1 + \gamma \frac{\ell_s(d_i)}{\ell_s(d)}}\right]\left[\prod_{\Phi_{\rm m}}\frac{1}{1 + \gamma \frac{P_{\rm s}}{P_{\rm r}}\frac{\ell_m(d_j)}{\ell_s(d)}}\right], 
    \end{align}
    where $\Phi'_{\rm r}$ is the reduced point process $\Phi_{\rm r} \backslash \{\bf x_1\}$, i.e., the resultant RN process by removing the serving RN.
    \label{lem:CSP}
\end{lemma}
\begin{IEEEproof}
This is a standard result in \ac{SG}, e.g., see~\cite{haenggi2021meta} for the general result and \cite{wang2018sir} for the multi-tier case.
\end{IEEEproof}
The \ac{CCDF} of the conditional success probability is called the \ac{MD}.
\begin{align}
    F_{P_{\rm s}}(\beta, \gamma) = \mathbb{P}\left(P_{\rm s}(\gamma) > \beta\right).
\end{align}
{In stationary and ergodic networks, the \ac{MD} $F_{P_{\rm s}}(\beta, \gamma)$ can be interpreted 
as the fraction of time that a typical UE experiences an SINR above $\gamma$ with reliability at least $\beta$. However, in our bounded BPP + PPP setting, ergodicity does not hold. Therefore, $F_{Ps}(\beta, \gamma)$ should be interpreted as the \emph{fraction of UEs across all network realizations} (of RN and UE locations) for which the conditional success probability exceeds $\beta$ at threshold $\gamma$. This distinction is particularly relevant for network planning: in enterprise deployments, multiple realizations correspond to different UE location configurations and possible RN placements. Hence, the MD provides a statistically meaningful measure of reliability across scenarios, enabling planners to dimension the network so that a target fraction of UEs achieves a required reliability level.}

\begin{rmk}
    In the case of coverage dimensioning in terms of \ac{SINR} or \ac{RSRP} rather than capacity dimensioning, the value of $\gamma$ can be directly set as the minimum \ac{SINR} required for decoding the control channel or the minimum target \ac{RSRP} in the downlink.
\end{rmk}
In what follows, we proceed with a fixed $\gamma$ and discuss the implications of cell load in a separate result in Section~\ref{sec:NRD}. Deriving the \ac{MD} of the SINR directly is very likely infeasible~\cite{haenggi2021meta}. In literature, researchers often characterize the \ac{MD} through the moments of the conditional success probability, i.e., $\mathbb{E}\left[P^b_{\rm s}(r_0)\right]$. This requires the evaluation of the\ac{PGF} of the underlying point process. For $\Phi_{\rm r}$, this is particularly challenging due to its asymmetry. First, we need the distribution of the distance of the test UE from the serving RN as follows.
\begin{lemma}
\label{lem:nearestRNdist}
~\cite[Th.~2.3.6]{mathai1999introduction}
    The distribution of the distance from the test UE to the nearest RN $D_{\rm r}$, conditioned on $r_0$ is
    \begin{align}
        &f_{D_{\rm r}}(y| r_0) = \nonumber \\
        &\begin{cases}
            \frac{2y}{R^2}; &0 \leq y \leq R - r_0 \nonumber \\
            \frac{2y}{\pi R^2} \cos^{-1} \left(\frac{-R^2 + y^2 + r_0^2}{2r_0y}\right); &R - r_0 \leq y \leq R + r_0. \nonumber 
        \end{cases}
    \end{align}
\end{lemma}
The next result that we need is the accurate characterization of the $b-$th moment discussed below.
\begin{theorem}
\label{theo:Moment_RN}
    The $b-$th moment of the conditional success probability $P_{\rm s}$ is given as
    \begin{align}
        M_b(r_0) = \int T_N(r_0, y) T_R(r_0, y) T_M(r_0, y) f_{D_{\rm r}}(y|r_0)dy,  \nonumber 
    \end{align}
    where, the conditional distribution $f_{D_{\rm r}}(y|r_0)$ is given in Lemma~\ref{lem:nearestRNdist}, $T_N(r_0, y)$ corresponds to the impact of the noise power. $T_R(r_0, y)$ corresponds to the impact of interference from non-serving RNs, and $T_M(r_0, y)$ corresponds to the impact of of the interference from the \ac{MBS} tier. In particular,
    \small{
    \begin{align}
        T_N(r_0, y) &= \exp\left(\frac{-b\gamma N_0}{KP_{\rm r}\ell_s(y)}\right), \nonumber \\
        T_R(r_0, y) &= \left(\frac{1}{2\pi R} \int_{0}^{\pi} \int_{r_1(\theta)}^{r_2(\theta)} \left(\frac{1}{1 + \gamma \frac{\ell_s(d(r, \theta))}{\ell_s(y)}}\right)^b {\rm d}r {\rm d}\theta + \right. \nonumber \\ 
    &\left. \int_{r_3(\theta)}^{r_4(\theta)} \left(\frac{1}{1 + \gamma \frac{\ell_s(d(r, \theta))}{\ell_s(y)}}\right) {\rm d}r {\rm d}\theta\right)^{N_S - 1}, \nonumber \\ 
        T_M(r_0, y) &= \exp\left(-\lambda_M \int_0^{2\pi}\int_{R}^{R_M} 1 -  \right. \nonumber \\
        &\left.\left(\frac{1}{1 + \gamma \frac{\ell_m\left(d(r, \theta)\right)}{\ell_s(y)}}\right)^b {\rm d}r {\rm d}\theta \right),
    \end{align}
    }
    \normalsize
    where $d(r) = r^2 + r_0^2 + 2rr_0 \cos \theta$. The limits $r_1(\theta)$ and $r_2(\theta)$ are respectively:
    \begin{align}
        r_1(\theta) &= \min\{-R, r_0-\sqrt{y^2\sec^2\theta - r_0^2\tan^2\theta}\},\nonumber \\
        r_2(\theta) &= \max\{-R, r_0-\sqrt{y^2\sec^2\theta - r_0^2\tan^2\theta}\}\nonumber, \\
        r_3(\theta) &= \min\{R, r_0+\sqrt{y^2\sec^2\theta - r_0^2\tan^2\theta}\}, \nonumber \\
        r_4(\theta) &= \max\{R, r_0+\sqrt{y^2\sec^2\theta - r_0^2\tan^2\theta}\}. \nonumber
    \end{align}
\end{theorem}
\begin{IEEEproof}
The $b$-th moment is evaluated as
\begin{align}
        M_b &= \mathbb{E}\left[P_{\rm r}^b\right]  = T_N(r_0, b) \cdot\nonumber \\
        &  \mathbb{E}\left[\prod_{\Phi_{\rm r}} \left(\frac{1}{1 + \gamma \frac{\ell_s(d_i)}{\ell_s(d)}}\right)\right] \mathbb{E}\left[\prod_{\Phi_{\rm m}}\left(\frac{1}{1 + \gamma \frac{P_{\rm s}}{P_{\rm r}}\frac{\ell_m(d_j)}{\ell_s(d)}}\right)^b\right] \nonumber
\end{align}
where $T_N(r_0, b, y) = \exp\left(\frac{-b\gamma N_0}{KP_{\rm r}\ell_s(y)}\right)$.
Now,
\begin{align}
    &\mathbb{E}\left[\prod_{\Phi_{\rm r}} \left(\frac{1}{1 + \gamma \frac{\ell_s(d_i)}{\ell_s(d)}}\right)\right] \nonumber \\
    &\overset{(a)}{=} \mathbb{E}_{\Phi_{\rm r}}\left[\frac{1}{2\pi R} \int_{0}^{2\pi} \int_0^{r_1(\theta)} \left(\frac{1}{1 + \gamma \frac{\ell_s(d(r, \theta))}{\ell_s(y)}}\right) {\rm d}r {\rm d}\theta + \right. \nonumber \\ 
    &\left. \int_{r_1(\theta)}^{r_2(\theta)} \left(\frac{1}{1 + \gamma \frac{\ell_s(d(r, \theta))}{\ell_s(y)}}\right) {\rm d}r {\rm d}\theta\right], \nonumber \\
    &\overset{(b)}{=} \left(\frac{1}{2\pi R} \int_{0}^{2\pi} \int_{r_1(\theta)}^{r_2(\theta)} \left(\frac{1}{1 + \gamma \frac{\ell_s(d(r, \theta))}{\ell_s(y)}}\right)^b {\rm d}r {\rm d}\theta \right)^{N_S - 1}, \nonumber
\end{align}
where $d(r) = r^2 + r_0^2 + 2rr_0 \cos \theta$. In step (a) the limits on the range of the distance of the interfering RNs are dependent on $\theta$. Step (b) follows from the fact that conditioned on the location of the nearest RN ${\bf x}_1$, the remaining point process $\Phi_{\rm r}$ is a BPP with $N_S -1$ points uniformly distributed in the area $\mathcal{A} = \mathbb{B}(0, R) \backslash \left(\mathbb{B}(0, R) \cap \mathbb{B}(r_0, d)\right)$. Thus, given that the interfering RNs are not present in $\mathcal{A}$, their distances are found by simultaneously solving $y = x \tan \theta$ and $(x - r_0)^2 + y^2 = d^2$, where the first equation is that of a line through the origin at an angle $\theta$ with the x-axis, while the second equation is that of a circle with the test UE at its center and a radius equal to the distance of the nearest RN.
In case the above equations do not have a solution for a given $\theta$, then at that angle, the RNs can have any distance between 0 to $R$. On the other hand, corresponding to the MBSs, we have
\begin{align}
    &\mathbb{E}\left[\prod_{\Phi_{\rm m}}\left(\frac{1}{1 + \gamma \frac{P_{\rm s}}{P_{\rm r}}\frac{\ell_m(d_j)}{\ell_s(d)}}\right)^b\right] \nonumber \\
    &= \exp\left(-\lambda_M \int_0^{2\pi}\int_{R}^{R_M} 1 - \left(\frac{1}{1 + \gamma \frac{\ell_m\left(d(r, \theta)\right)}{\ell_s(y)}}\right)^b {\rm d}r {\rm d}\theta \right). \nonumber
\end{align}
This follows from the \ac{PGF} of a \ac{PPP}.
\end{IEEEproof}
Using Lemma~\ref{lem:nearestRNdist} and Theorem~\ref{theo:Moment_RN}, we will characterize the \ac{MD}. However, before discussing the same, let us first analyze the impact of \ac{MBS} cooperation on the above.

\subsection{\ac{MBS} Cooperation}
In the presence of nearby \ac{MBS} of the same operator, cooperative deployment is necessary~\cite{ghatak2021coverage}. In particular, the \acp{UE} receiving sufficient downlink power from the outdoor \ac{MBS} may choose to connect to the \ac{MBS} instead of an indoor \ac{RN}. Furthermore, the location of the indoor \acp{RN} must be optimized by considering the interference caused by the \acp{RN} to such \acp{UE}. The conditional success probability must be evaluated by conditioning not only on the realization of the \ac{RN} and the \ac{MBS} process but also on the association state, i.e., either \ac{RN} or \ac{MBS} association. Note that for the multi-tier case, the derivation of the \ac{SINR} \ac{MD} is significantly more challenging compared to the analysis present in~\cite{wang2018sir} and \cite{cui2017sir} due to the non-stationarity of $\Phi_{\rm s}$ and $\Phi_{\rm m}$. First, let us consider the distance to the nearest \ac{MBS} from the test \ac{UE}, denoted by $D_{\rm M}$. As shown in Fig.~\ref{fig:nearest_MBS}, we need to find the probability that there are no MBS present in the light blue region. The distribution is derived below from simple geometric arguments and hence the proof is being skipped.
\begin{figure}
    \centering
    \includegraphics[width = 0.6\linewidth]{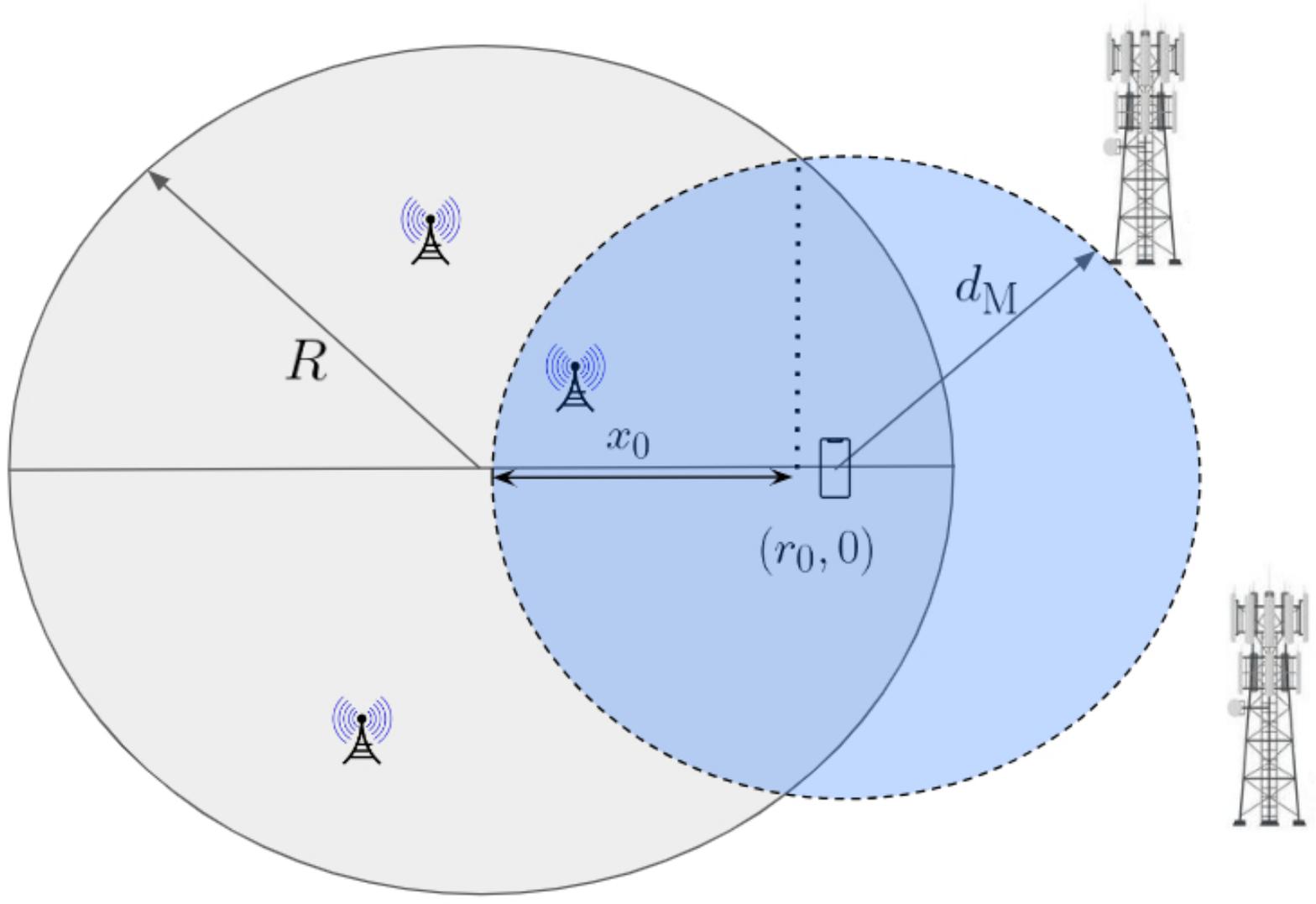}
    \caption{Illustration of the nearest \ac{MBS} distance.}
    \label{fig:nearest_MBS}
\end{figure}
\begin{lemma}
\label{lem:nearestMBSdist}
The distribution of the distance from the test UE to the nearest MBS, conditioned on $r_0$ is given by
    \begin{align}
        f_{D_{\rm m}}(y | r_0) = \begin{cases}
            f_1(y,r_0); -R \leq x' \leq R ,  \nonumber \\ 
2 \pi y \lambda_M \exp\left(-\pi\lambda_M\left(y^2- R^2\right)\right); \text{otherwise},
        \end{cases}
    \end{align}
    where $f_1(y, r_0)$ is given in ~\eqref{eqpdfd1} and $x'$ is the abscissa of the intersection point of the two circles.
    \begin{table*}
\centering
        \small{
    \begin{align}\nonumber
 f_{d_{\rm M}}(y, r_0)=&\lambda_{\rm M} exp\left(-\lambda_{\rm M} \left(y^2\left(\frac{\pi}{2}+t_2\sqrt{1-t_2^2}-\tan^{-1}\frac{-t_2}{\sqrt{1-t_2^2}}\right) -R^2\left(\frac{\pi}{2}-t_1\sqrt{1-t_1^2}-\tan^{-1}\frac{t_1}{\sqrt{1-t_1^2}}\right)\right)\right)\\
 &\times\frac{-yR}{r_0}\left[{2\sqrt{1-t_1^2}}\right]++t_3y^2\left[2{\sqrt{1-t_2^2}}\right]+2y\left[t_2\sqrt{1-t_2^2}+\frac{\pi}{2}-\tan^{-1}\frac{-t_2}{\sqrt{1-t_2^2}}\right]
 \label{eqpdfd1}
\end{align}
}
\normalsize
where $t_1=\frac{R^2y^2+r_0^2}{2r_0R}$, $t_2=\frac{r_0^2-R^2+y^2}{2r_0y}$, and $t_3=\frac{y^2+r_d^2-u_0^2}{2u_0y^2}$.
    \hrule
\end{table*}
\end{lemma}
Given this distribution and Lemma~\ref{lem:nearestRNdist}, the conditional association probabilities are derived below.
\begin{lemma}
Conditioned on $r_0$, the access probabilities (alternatively, the probability of dominance over macro) for the \ac{UE} are
\begin{align}
    \mathcal{P}_{\rm r}(r_0) &= \mathbb{E}\left[F_{D_{\rm r}}\left( \left(\frac{KP_{\rm m}}{KP_{\rm r}}\right)^{-\frac{1}{\alpha_{\rm r}}}\right) D_{\rm m}^{\frac{\alpha_{\rm m}}{\alpha_{\rm s}}}\right], \nonumber \\
        \mathcal{P}_{\rm m} &= 1 - \mathcal{P}_{\rm r}(r_0), \nonumber
\end{align}
where $F_{D_{\rm r}}$ is the CDF of the nearest RN (the measure with respect to the PDF given in \eqref{lem:nearestRNdist}) while the expectation is with respect to the distribution given in Lemma~\ref{lem:nearestMBSdist}.
\end{lemma}
\begin{IEEEproof}
    The proof follows directly by evaluating the probability $\mathbb{P}\left(KP_{\rm r}D_{\rm r}^{-\alpha_{\rm r}} > KP_{\rm m}D_{\rm m}^{-\alpha_{\rm m}}\right)$ for random variables $D_{\rm r}$ and $D_{\rm m}$.
\end{IEEEproof}
Now, unlike typical \ac{SG} analysis of the standard success probability, we define the {\it association-conditioned} success probability as a discrete random variable conditioned on realizations of the underlying point processes $\Phi_{\rm m}$ and $\Phi_{\rm r}$ as:
\begin{align}
    P_{\rm s}(r_0) = \begin{cases}
        P_{{\rm s}|{\rm r}}(r_0) &\text{with probability } \mathcal{P}_{\rm r}(r_0), \nonumber \\
        P_{{\rm s}|{\rm m}}(r_0) &\text{with probability } \mathcal{P}_{\rm m}(r_0),
    \end{cases}
\end{align}
where
\begin{align}
    &P_{{\rm s}|{\rm r}}(r_0) =  \nonumber \\
    &e^{\left(\frac{-\gamma N_0}{KP_{\rm r}\ell_s(d_S)}\right)} \cdot \nonumber \left[\prod_{\Phi_{\rm r}}\frac{1}{1 + \gamma \frac{\ell_s(d_i)}{\ell_s(d)}}\right]\left[\prod_{\Phi_{\rm m}}\frac{1}{1 + \gamma \frac{P_{\rm s}}{P_{\rm r}}\frac{\ell_m(d_j)}{\ell_s(d)}}\right].
\end{align}
and
\begin{align}
    &P_{{\rm s}|{\rm m}}(r_0) = \nonumber \\
    &e^{\left(\frac{-\gamma N_0}{KP_{\rm s}\ell_M(d_M)}\right)}
    \left[\prod_{\Phi_R}\frac{1}{1 + \gamma \frac{\ell_s(d_i)}{\ell_s(d)}}\right]\left[\prod_{\Phi'_M}\frac{1}{1 + \gamma \frac{P_{\rm s}}{P_{\rm r}}\frac{\ell_m(d_j)}{\ell_s(d)}}\right].
\end{align}
The first moment of the conditional success probability gives us the standard success probability, which is often referred to as the coverage probability.
\begin{corollary}
The standard success probability can be obtained by taking an expectation of $P_{\rm s}(r_0)$ over the realizations of $\Phi_{\rm m}$ and $\Phi_{\rm s}$ and it matches the classical SINR coverage probability for multi-tier networks. Mathematically, it is
    \begin{align}
        p_{\rm s}(r_0) &= \mathbb{E}_{\Phi_{\rm s}, \Phi_{\rm m}}\left[P_{\rm s}(r_0)\right] \nonumber \\
        &= p_{\rm r}(r_0) \mathcal{P}_{\rm r}(r_0) + p_{\rm m}(r_0) \mathcal{P}_{\rm m}(r_0),
    \end{align}
    where
    \begin{align}
         p_{\rm r}(r_0) = \mathbb{E}^0_{\rm r}\left[P_{{\rm s}|{\rm r}}(r_0)\right]; \;\;
         p_{\rm m}(r_0) = \mathbb{E}^0_{\rm m}\left[P_{{\rm s}|{\rm m}}(r_0)\right]. \nonumber
    \end{align}
$\mathbb{E}^0_{\rm r}$ and $\mathbb{E}^0_{\rm m}$ respectively refer to the conditional point processes given \ac{RN} and \ac{MBS} association under the reduced Palm conditioning~\cite{chiu2013stochastic}.
\end{corollary}

\begin{corollary}
    The $b-$th moment is obtained in a similar manner
    \begin{align}
        M_b(r_0) = M_{b|R}(r_0) \mathcal{P}_R(r_0) + M_{b|M}(r_0) \mathcal{P}_M(r_0),  \nonumber
    \end{align}
    where $M_{b|t} = \mathbb{E}\left[P^b_{S|t}\right]$ for $t \in \{R,M\}$.
\end{corollary}
\begin{rmk}
    Albeit intuitive, a word of caution is necessary while working with \ac{MD}s of multi-tier networks. The overall conditional success probability for a given realization is not $P_{{\rm s}|{\rm r}} \mathcal{P}_{\rm r}(r_0) + P_{{\rm s}|{\rm m}} \mathcal{P}_{\rm m}(r_0)$ and hence, $M_b(r_0)  \neq \mathbb{E}\left[\left(P_{{\rm s}|{\rm r}} \mathcal{P}_{\rm r}(r_0) + P_{{\rm s}|{\rm m}} \mathcal{P}_{\rm m}(r_0)\right)^b \right]$.
\end{rmk}
Of special interest are $b = 1$ and $b = -1$. Specifically, $M_1(r_0)$ represents the standard success probability $P_{\rm r}(r_0)$, often referred to as the SINR coverage probability. $M_{-1}$ represents the mean-local delay, i.e., the average number of re-transmissions required for success in the UE-RN link.
The exact \ac{MD} of the SINR can be calculated using the Gil-Palaez theorem as
\begin{align*}
    F_{P_{\rm r}}(z) = \frac{1}{2} + \frac{1}{\pi}\int_0^\infty \frac{\Im(e^{-ju \log(z)}M_{ju})}{u} \mathrm{d}u,
\end{align*}
where $\Im\left(\cdot\right)$ is the imaginary part of the argument and $M_{ju}$ is the $ju$-th moment. Often a simple yet accurate approximation of the SINR MD is useful leveraging the beta distribution. This approximation arises naturally as the support of the \ac{MD} is $[0, 1]$, and is obtained by matching the first moment $M_1$ and the second moment $M_2$. Specifically, the MD is approximately expressed as
\begin{align*}
    \bar{F}(\beta, t) \approx 1 - I_{\varepsilon}\left(\frac{k_1 k_2}{1-k_1}, k_2\right),
\end{align*}
where $I_{\varepsilon}(x,y) = \int_0^{1-\varepsilon} z^{x-1}(1-z)^{y-1}\mathrm{d}z/B(x,y)$ is the regularized incomplete beta function with $B(\cdot, \cdot)$ the beta function, $k_1 = M_1$, and $k_2 = (M_1 - M_2)(1-M_1)/(M_2 - M_1^2)$. The accuracy of the beta approximation is confirmed in Fig.~\ref{fig:md_basic}.
{It must be emphasized that the \ac{SG} analysis, based on the \ac{BPP} assumption, corresponds to a scenario where the RNs are uniformly and randomly placed within the deployment area. As such, the number of RNs estimated via SG should be interpreted as a \emph{conservative upper bound} on the true optimal number of RNs. Optimized placement strategies that account for the environment and user distribution will, in general, require fewer RNs to meet the same performance target. Therefore, the SG-based estimate serves as a useful \emph{baseline for unplanned deployments}, while the subsequent ILP and clustering-based approaches refine this bound toward an efficient and practically realizable network design. To summarize, at this stage, the \ac{MD} analysis provides, for each candidate RN density, the fraction of UEs across network realizations that meet the reliability target $(\beta, \gamma)$. Given a network-level design requirement (e.g., 90\% of UEs achieving SINR $\geq \gamma$ with reliability $\beta$), the minimum RN density that satisfies this target is identified. This determines the corresponding number of RNs, denoted by $K$, which serves as the baseline for subsequent clustering and refinement in Section~V.}
\begin{table*}[t]
\centering
\begin{tabular}{|l|p{0.35\linewidth}|p{0.35\linewidth}|}
\hline
                       & ILP & MD \\ \hline
Accuracy      & Local optimal for a given environment  &Probabilistic   \\ \hline
Generalizability&Non-generalizable, insights not transferable &Distribution across all network realizations\\
\hline
Small-Scale Fading &Not taken into account &Statistical models (Rayleigh, Nakagami, etc) \\
\hline
Impact of blockages &Accurate characterization & Statistical models (3GPP etc.) \\
\hline
Computation complexity &Very high, often infeasible &Low \\
\hline
\end{tabular}
\caption{Difference between MD analysis and ILP.}
\label{table:comp_table_MD_ILP}
\end{table*}

{The \ac{SG} derivation adopts several simplifying assumptions for analytical tractability: (i) an isotropic, disk-shaped deployment region, (ii) uniformly distributed UEs, (iii) a distance-only pathloss model $\ell(d)=K d^{-\alpha}$ with no per-wall deterministic penetration term, and (iv) small-scale fading (Rayleigh) but no closed-form shadowing model. These assumptions enable closed-form expressions for the complex moments and a beta-distribution approximation for the meta-distribution, but they also limit the direct numerical accuracy of the expressions in environments that strongly deviate from these assumptions.

In particular, the closed-form angular integrations in Theorem~1 exploit isotropy of the disk domain. For irregular domains or severe edge effects, the angular limits change and the integrals must be evaluated numerically or replaced by Monte-Carlo estimates. Furthermore, if the UE locations follow a general density $\lambda_u(x)$, any empirical spatial average $\frac{1}{N}\sum_i g(x_i)$ in the moment computations must be replaced by the spatial integral $\frac{1}{|A|}\int_A g(x)\lambda_u(x)\,dx$ (or weighted average with $\lambda_u$). The MD framework itself remains valid; only the integrals that produce $M_b$ change.
  
\begin{rmk}
Including a multiplicative shadowing factor $S$ (e.g., log-normal) that multiplies the pathloss leads to a modification of the $b$-th moment expressions by the factor $\mathbb{E}[S^{-b}]$ wherever $\ell(\cdot)$ appears in denominators inside Laplace/PGF terms. For example, replacing $\ell_s(y)$ by $\ell_s(y)S$ changes the noise term $\exp\!\big(-b\gamma N_0/(K P_r \ell_s(y))\big)$ to its shadowed expectation, and the interference contribution from a tier obtains multiplicative moment-corrections of the form $\mathbb{E}[S^{-b}]$.    
\end{rmk}
}

\subsection{Estimating the Number of \acp{RN} via Integer Linear Program}
While the \ac{SG} analysis gives a statistical characterization for a given deployment area dimension, the downlink \ac{RSRP} and SINR requirements can directly be formulated as an RN minimization problem. Let us approximate the deployment area as a square grid consisting of $n$ cells, $g_i$, $i = \{1, 2, \ldots, n\}$. Let $o_i$ represent the RN occupancy of cell $g_i$, i.e., $o_i = 1$ if an RN is place in cell $g_i$, otherwise we set $o_i = 0$. Furthermore, let us define the connectivity matrix ${\bf C} = \{0,1\}_{\sqrt{n} \times \sqrt{n}}$, wherein, $c_{ij} = 1$ if the cell $g_i$ associates to an RN in cell $g_j$. Furthermore, let us define the symmetric power matrix ${\bf A}_{\sqrt{n} \times \sqrt{n}}$ where $a_{ij}$ represents the fading-averaged downlink power (assuming reciprocity) between the cell $g_i$ and $g_j$, i.e., $a_{i,j} = K P ||g_i - g_j||^{-\alpha}$.
\begin{subequations}
\begin{align}
    \text{min} \;\; & \sum_i^n o_i, \label{eq:obj}\\
    \text{s t} \;\; & \sum_i c_{ij} = 1, \;\; \forall j, \label{const1} \\
    & c_{ij} - o_j \leq 0, \label{const2}\\
    & K_\infty \left(1 - z_{ij}\right) + a_{ij}o_j \geq N'_{ij} + \sum_{j \in \mathcal{G}} c_{ij} o_j. \label{const3}
\end{align}
\label{eq:ilp}
\end{subequations}
Thus, the optimization variables are ${\bf o}_{n \times 1}$ and ${\bf C}_{n \times n}$, that can be stacked together as a single vector ${\bf v} = [{\bf o}^T \text{vec}({\bf C})^T]^T$. The objective, \eqref{eq:obj} denotes that we are interested in minimizing the number of RNs. \eqref{const1} forces each cell/UE to connect to exactly one RN. The constraint \eqref{const2} denotes that a UE in cell $g_i$ can connect to cell $g_j$ iff an RN is placed in $g_j$. Finally, the constraint \eqref{const3} satisfies the minimum SINR requirements for each UE, i.e., $\xi_i \geq \gamma$. $K_\infty$ is a large value that ensures that we are interested in the SINR term only when a UE in $g_i$ connects to an RN in $g_j$, otherwise we ignore the constraint. In order to solve \eqref{eq:ilp}, we reformulate it into the following program.
\begin{subequations}
\begin{align}
    \text{min} \;\; & {\bf f}^T{\bf v},  \nonumber \\
    \text{s t} \;\;& {\bf A}_{eq}{\bf v} = {\bf b}_{eq}; \quad \nonumber  {\bf A}{\bf v} \leq {\bf b}. \nonumber 
\end{align}
\label{eq:ilp1}
\end{subequations}
where,
\begin{align}
    {\bf A}_{eq} &= \left[{\bf 0}_{n \times n} \;\; {\bf 1}^T \otimes {\bf I}_{n \times n}\right], \quad
    {\bf b}_{eq} = {\bf 1}_{n \times 1}, \nonumber \\
    {\bf A} &= \left[{\bf A}_1^T {\bf A}_2^T\right]^T, \quad {\bf b} = \left[{\bf b}_1^T \;\; {\bf b}_2^T\right]^T, \nonumber \nonumber \\
    {\bf A}_1 &= \left[\left({\bf I}_{n \times n} \otimes -2{\bf A}\right)\cdot \chi \;\; K_\infty{\bf I}_{n \times n}\right] + \left[{\bf 1}\otimes {\bf A} \;\; {\bf 0}_{n^2 \times n^2}\right], \nonumber \\
    {\bf A}_2 &= \left[{\bf I}_{n \times n} \otimes -{\bf 1}_{n \times n} \;\; {\bf I}_{n^2 \times n^2}\right], \nonumber \\
    {\bf b}_1 &= {\bf 0}_{n^2 \times 1}, \quad
    {\bf b}_2 = \left(-\gamma + K_\infty\right){\bf 1}_{n^2 \times 1}. \nonumber 
\end{align}
In the above $\otimes$ represents the Kronecker product. We solve the above using the Gurobi optimizer version 11.0.0~\cite{anand2017comparative}. The solution of the ILP often reaches a local minima without providing much insights into the placement problem. However, as we shall see in Section~\ref{sec:NRD}, this is a useful tool if we need to model exact blockage environments and results in a more accurate estimation of the number of RNs as compared to the SG approach. {In essence, the ILP formulation jointly optimizes the binary RN placement vector $o\in\{0,1\}^n$ (one entry per candidate cell) and the connectivity matrix $C\in\{0,1\}^{n\times n}$ (assignment of each cell to exactly one RN). For a discretization consisting of $n$ candidate cells, the total number of binary decision variables in the stacked vector $v=[o^T,\operatorname{vec}(C)^T]^T$
scales as $\mathcal{O}(n^2)$. The number of linear constraints (assignment constraints, coupling constraints, and SINR constraints) also scales on the order of $O(n^2)$. These scalings directly imply that the ILP size grows quadratically with the grid resolution (number of candidate cells). Integer programming is NP-hard in general and thus the worst-case time complexity is exponential in the number of binary variables; modern solvers (e.g., branch-and-bound / branch-and-cut) often solve many practical instances far more quickly than the worst case, but no polynomial-time
guarantee exists.}
Nevertheless, Table~\ref{table:comp_table_MD_ILP} summarizes the advantages and the disadvantages of a statistical approach such as \ac{SG} with an optimization-based approach such as \ac{ILP}. Apart from the convergence to a local minima, the ILP formulation has other issues, e.g., no consideration of the fading conditions, and ignoring the UE density in the deployment area. 

{
\subsection{Deployability-aware extension of the placement ILP} We further extend the framework to take into account realistic deployment scenarios. Let $\mathcal{S}$ be the candidate RN sites, $\mathcal{U}$ the demand points, and $\mathcal{H}$ the wiring closets or PoE switches. For each site $s$, let $f_s\in\{0,1\}$ indicate structural and maintenance feasibility and let $e_s\in\{0,1\}$ indicate power availability. A cable-routing graph is used to precompute the shortest feasible route length $L_{sh}$ from site $s$ to closet $h$, and $r_{sh}=\mathbf{1}\{L_{sh}\leq L_{\max}\}$ indicates whether the route satisfies the PoE-reach limit. Let $C_h$ denote the available port capacity at closet $h$, $c_s^{\mathrm{inst}}$ the site-dependent installation cost, and $c^{\mathrm{cab}}$ the cable cost per unit route length.

The binary variables $o_s$, $c_{us}$, and $z_{sh}$ indicate RN installation, demand-point association, and site-to-closet connection, respectively. The baseline placement problem can be extended as
\begin{align*}
\min_{\{o,c,z\}}\quad
& \sum_{s\in\mathcal{S}}c_s^{\mathrm{inst}}o_s
+c^{\mathrm{cab}}\sum_{s\in\mathcal{S}}\sum_{h\in\mathcal{H}}L_{sh}z_{sh} \\
\mathrm{s.t.}\quad
& \sum_{s\in\mathcal{S}}c_{us}=1, & \forall u\in\mathcal{U},\\
& c_{us}\leq o_s, & \forall u\in\mathcal{U},\ s\in\mathcal{S},\\
& o_s\leq f_s,\qquad o_s\leq e_s, & \forall s\in\mathcal{S},\\
& \sum_{h\in\mathcal{H}}z_{sh}=o_s, & \forall s\in\mathcal{S},\\
& z_{sh}\leq r_{sh}, & \forall s\in\mathcal{S},\ h\in\mathcal{H},\\
& \sum_{s\in\mathcal{S}}z_{sh}\leq C_h, & \forall h\in\mathcal{H},\\
& \xi_u\geq\gamma, & \forall u\in\mathcal{U}.
\end{align*}
Here, the final constraint uses the same linearized RSRP/SINR expression as the baseline ILP. A hard deployment budget can additionally be imposed as
\[
\sum_s c_s^{\mathrm{inst}}o_s+c^{\mathrm{cab}}\sum_{s,h}L_{sh}z_{sh}\leq C_{\max}.
\]
The formulation treats radio feasibility and physical deployability jointly: prohibited or inaccessible sites are excluded through $f_s$, unavailable power locations through $e_s$, PoE reach through $r_{sh}$, closet capacity through $C_h$, and cable effort through the objective or budget. If conduit-level capacities must be represented, the precomputed route variables can be replaced by a standard edge-flow submodel without changing the remainder of the framework.
}

\section{Initial Placement of \acp{RN}}
\label{sec:IP}
{Based on the MD and/or ILP analysis described in Section~IV, we first determine the number of RNs $K$ that satisfies the target reliability $(\beta, \gamma)$. This $K$ serves as the input to the clustering framework described next} where the task is to identify the initial locations for the $K$ \acp{RN}. For the same, we consider a prior knowledge about the distribution of the UE locations. Specifically, let there be N UEs in the network with locations $\Phi_{\rm u} = {\bf x_i}$, $i = 1, \ldots, N$ as defined in Section~II. Furthermore, let the distribution of the location of each ${\bf x}_i$ be independent and identically distributed. In this section, we discuss RN placement startegies for a given UE location realization. Then, in Section~\ref{sec:NRD}, we carry out a statistical analysis of the same over both the UE and RN realizations.
\subsection{KM and CKM}
The classical \ac{KM} clustering algorithm~\cite{macqueen1967classification,lloyd1982least} minimizes the sum of the squared Euclidean distance from the \acp{UE} to their nearest \acp{RN}. Mathematically, this is represented as
\begin{align}
    \mathcal{M}_{\rm KM} = \sum_{i = 1}^N \min_l\{||x_i - m_l||^2\}. \label{eq:KM metric}
\end{align}
For uniform locations of \acp{UE}, \ac{KM} intuitively appears to be a reasonable method to designate the \ac{RN} locations. However, for a non-homogeneous \ac{UE} density, it is well-known that \ac{KM} may result in imbalanced cluster allotments.

To alleviate the above issue, \cite{bradley2000constrained} had proposed the \ac{CKM} algorithm. \ac{CKM} and its variants impose an upper and/or a lower limit on the fraction of \acp{UE} that associate to each \ac{RN}. Formally, \ac{CKM} minimizes \eqref{eq:KM metric} subject to
\begin{align}
    N_{\rm min} \leq  N_l \leq N_{\max} \label{eq:capacity constraint}. 
\end{align}
where $0 \leq N_{\min}, N_{\max} \leq N$. However, such a capacity constraint albeit relevant for several applications, is unsuitable for many wireless network realizations as will be shown in Section~\ref{sec:NRD}. This is precisely due to the fact that in order to meet the capacity constraint, \ac{CKM} tends to create disproportionately sized clusters, i.e., larger cluster area in regions with lower \ac{UE} density while smaller cluster area in regions with higher \ac{UE} density. This results in a higher number of cell-edge \acp{UE} of large cells thereby deteriorating their signal strength.

\subsection{\ac{KHM}: Optimal for Fully Clustered \acp{RN}}
Although \ac{KM} and \ac{CKM} are intuitive algorithms and are simple to implement, for finite \ac{UE} locations, they do not optimize the downlink \ac{UE} signals. \ac{KM} and \ac{CKM} are not guaranteed to be sum-\ac{RSRP} optimal for finite number of UEs. 
In this regard, the \ac{KHM} algorithm, which facilitates inverse weightage to the Euclidean distances is an attractive choice. Specifically, the \ac{KHM} metric as defined in~\cite{zhang1999k, zhang2000generalized} is:
\begin{align}
    \mathcal{M}_{\rm KHM} = \sum_{i = 1}^{N_S}\frac{K}{\sum_{l = 1}^{K} \frac{1}{||{\bf x}_i - {\bf m_l}||^p}}.
\end{align}
Note that the \ac{KHM} metric considers the sum of the \ac{RSRP} from all the \acp{RN} to each \ac{UE}. In case the \acp{RN} are not clustered, and form separate cells, this formulation thus is not an accurate representation of the downlink \ac{UE} performance. In fact the $\mathcal{M}_{\rm KHM}$ formulation is counter-productive since it considers the interfering signals  However, for clustered cells, $\mathcal{M}_{KHM}$ is an exact representation of sum-\ac{RSRP}.
\begin{rmk}
    Minimizing $\mathcal{M}_{KHM}$ gives the optimal \ac{RN} positions for a completely clustered cell.
\end{rmk}

\begin{figure}
    \centering
    \includegraphics[width = 0.65\linewidth]{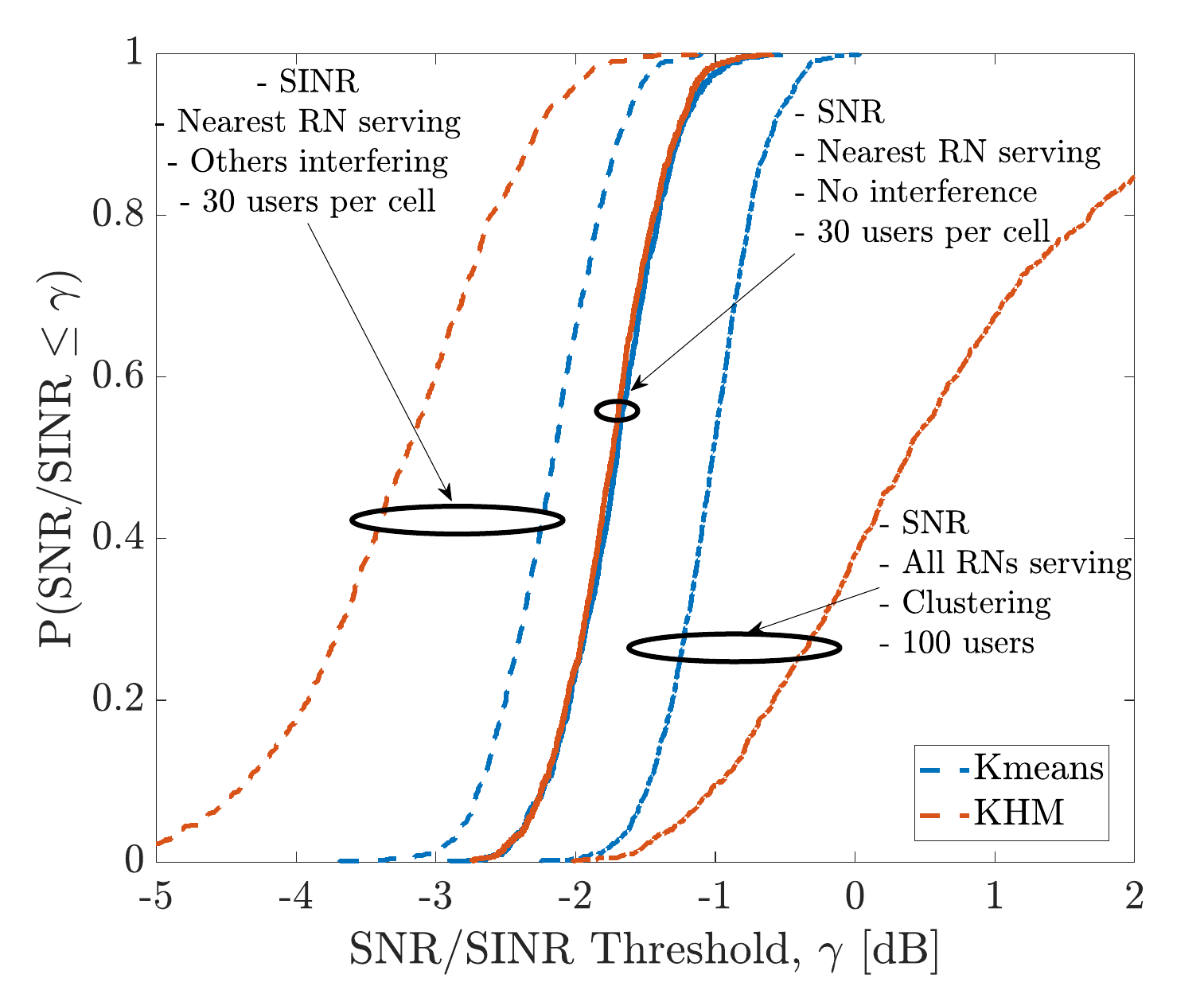}
    \caption{Comparison of the performance of K-means vs K-Harmonic Means in selecting the locations of the clustered RNs. Here $N_S = 3$.}
    \label{fig:KM_v_KHM}
\end{figure}
Fig.~\ref{fig:KM_v_KHM} considers 100 uniformly distributed UEs and 3 RNs. Without full clustering, KM achieves higher SINR than KHM, while both yield similar SNR. Under complete clustering, KHM significantly outperforms KM in downlink SNR, as all 100 UEs are served by a single cell instead of $\sim$30–33 UEs per cell. As discussed in Section~\ref{sec:Cluster}, KHM is optimal under full cooperation, whereas with limited or no clustering, the nearest RN is treated as the serving cell.

\subsection{\ac{WKHM}}
To leverage the benefits of \ac{KM} and \ac{KHM} simultaneously, we propose the \ac{WKHM} metric in this section. Recall that the issue with \ac{KHM} is the consideration of the received power from all RNs for each \ac{UE}. Let us consider a weighted version of \ac{KHM} as follows.
\begin{align}
    \mathcal{M}_{\rm WKHM} = \sum_{i = 1}^{N_S}\frac{K}{\sum_{l = 1}^{K} \frac{w_{il}}{||{\bf x}_i - {\bf m_l}||^p}}.
\end{align}
where $w_{il}$ is the weight given to the \ac{UE} $i$ and RN $l$ link. Thus, ideally $w_{il} = 1$ when the $l-$th RN is the nearest one from the $i-$th UE, otherwise $w_{il} = 0$. However, such a construction results in the WKHM metric being non-differentiable. Consequently, deriving the update steps similar to \ac{KHM} is challenging. As a result we introduce soft weights $w_{il} = \frac{||{\bf x}_i - {\bf m_l}||^{-q}}{\sum_{l = 1}^ K ||{\bf x}_i - {\bf m_l}||^{-q}}$ to characterize the WKHM metric. Note that the numerator and the denominator of the weights consist of the $-q$th power of the Euclidean distance. Thus the weight assigned corresponding to an RN for a UE is the ratio of the $-q$-th power of the Euclidean distance to that RN to the sum of the $-q-$th powers to the distances of all the RNs. Formally, this is called the signal fraction of the received signal. The derivative with respect to a cluster center is thus
\begin{align}
    \pdv{P}{m_l} = &\sum_{i = 1}^N \left[ \frac{K}{\left(\sum_{l = 1}^K d_{il}^{-q} \sum_{l = 1}^K \frac{1}{d_{il}^{p + q}}\right)}\left[\left(\sum_{l = 1}^K \frac{1}{d_{il}^{p+q}}\right) \right.\right. \nonumber \\
    &\left.\left. \left(q d_{il}^{-(q + 2)}\right) + \left(\sum_{l = 1}^K d_{il}^{-q}\right) \frac{p + q}{d_{il}^{p+q+2}}\right]\right]\left(x_i - m_l\right),
\end{align}
where $d_{il} = ||{\bf x}_i - {\bf m}_l||$. Equating the above to zero we obtain the following RN location update step:
$    m_l = \frac{\sum_i B_i x_i}{\sum_i B_i},$ 
where
\begin{align}
    B_i = \frac{1}{\left(\sum_{l = 1}^K d^{-q}_{il} \sum_{l = 1}^K d^{p+q}_{il}\right)^2} &\left[\left(\sum_{l = 1}^K \frac{-q}{d^{p+q}_{il}}\right)d^{-(q+2)}_{il} + \right. \nonumber \\
    &\left. d^{-q}_{il} \left(\frac{p + q}{d^{p+q+2}_{il}}\right) \right] \nonumber 
\end{align}
\begin{rmk}
    Note that for $w_{il} = 1, \forall i, l$, the WKHM update step coincides with the KHM update step as expected. 
\end{rmk}
{Unlike classical KHM, which minimizes a purely geometric harmonic-mean distance, the proposed WKHM replaces Euclidean distance with a radio-aware dissimilarity derived from RSRP/SINR statistics. As a result, the WKHM objective is not equivalent to KHM under any distance reweighting or constrained formulation. In particular, WKHM emphasizes worst-case users and lower-tail SINR performance, which is fundamentally different from the mean-distance minimization of k-means, the constrained variants of CKM, or the max-distance objective of k-centers.

The per-iteration computational complexity of WKHM remains $\mathcal{O}(N K)$, identical in order to KHM, with convergence ensured through monotonic descent of the modified objective. WKHM is therefore positioned as a radio-node location refinement mechanism tailored to indoor wireless planning, rather than as a generic replacement for existing clustering algorithms.}
As mentioned before, the network designer may not always dimension the network from the perspective of \ac{SINR} optimality. Often the industry standard is to design for a target \ac{RSRP} guarantee and study the impact on \ac{SINR} for such a dimensioning rule.

{
\noindent\textbf{Signal-fraction interpretation of WKHM.}
Let the slow-fading-averaged power received by UE $i$ from RN $l$ be
\(P_{il}=P_0\bigl(\lVert x_i-m_l\rVert^2+\epsilon\bigr)^{-\alpha/2}.
\)
For $q=\alpha$, the weights have the exact physical interpretation
\(
w_{il}=\frac{P_{il}}{\sum_{j=1}^{K}P_{ij}}.
\)
Let $S_i=\sum_jP_{ij}$ and let $v_i=\max_l P_{il}/S_i$ denote the serving-signal fraction under strongest-RN association. The SINR can then be written as
\[
\xi_i=\frac{v_iS_i}{(1-v_i)S_i+N_0}
     =\frac{v_i}{1-v_i+N_0/S_i}.
\]
For $0<v_i<1$ and $S_i>0$,
\[
\frac{\partial \xi_i}{\partial v_i}
=\frac{1+N_0/S_i}{(1-v_i+N_0/S_i)^2}>0,
\quad
\frac{\partial \xi_i}{\partial S_i}
=\frac{v_iN_0/S_i^2}{(1-v_i+N_0/S_i)^2}>0.
\]
Hence, increasing both the serving-signal fraction and aggregate received power is sufficient to increase the SINR. Moreover, for the regularized weighted-distortion representation and $p=q=\alpha$, the per-UE term becomes
\[
\sum_{l=1}^{K}w_{il}(\lVert x_i-m_l\rVert^2+\epsilon)^{p/2}
=\frac{KP_0}{S_i}.
\]
Because $x\mapsto 1/x$ is convex and decreasing on $\mathbb{R}_{+}$, this harmonic penalty emphasizes UEs with small aggregate received power. The normalized signal-fraction weights simultaneously make each RN update depend most strongly on UEs for which that RN is a plausible serving node. WKHM is therefore a tail-aware surrogate aligned with sufficient conditions for SINR improvement. It does not, however, constitute a global SINR-optimality guarantee, since a centroid displacement may also alter the interference geometry.
}

\subsection{\ac{KC}: Max-Min Fairness}
In case the network designer intends to guarantee for the worst case \ac{UE} \ac{RSRP} experience, the \ac{KC} algorithm~\cite{gonzalez1985clustering} can be an attractive option. The \ac{KC} algorithm attempts to maximize the minimum distance between the RN centers and the \acp{UE}. The corresponding metric is.
\begin{align}
    \mathcal{M}_{\rm KC} = \max_i \{ \min_l \{ ||x_i - m_l||^2\}\}.
    \label{eq:KC_metric}
\end{align}
Unfortunately, the \ac{KC} clustering problem is NP-HARD, and thus, it is unlikely that there can ever be efficient polynomial time exact algorithms. However, sub-optimal approximate algorithms exist for the problem. In our simulations, we use the algorithm presented in \cite{kleindessner2019fair} for solving the KC problem. { Regardless of the choice of the clustering algorithm, a note on their implementation is necessary. In fact, the final deployment based on the particular choice of an algorithm is not simply based on its evaluation in one particular realization, but rather averaged out over multiple realizations based on the UE distribution. This is important since individual realizations may vary in their cardinality. As an implementation example, consider that the UE distributions are exponentially skewed in both the abscissa and the ordinate. Then, for we generate several realizations of the UE locations following the same distribution and for each case, we obtain the algorithm-specific centroid locations. The final deployment is the mean centroid locations over the realizations. This indirectly takes into account the UE distribution while smoothening out the fluctuations of the result over individual realizations.}

{The clustering objectives adopted in this work are primarily 
based on RSRP or distance (e.g., KM, KHM, KC), which offer analytical tractability and directly 
align with coverage-based dimensioning targets commonly used in industry practice. RSRP-based metrics 
also enable scalable clustering algorithms without the need for iterative interference calculations. However, 
in strongly interference-limited regimes, SINR can be a more pertinent clustering criterion. The proposed 
framework is sufficiently flexible to incorporate SINR-based objectives if required. First, the WKHM clustering method attempts to incorporate SINR measurements by leveraging the signal fraction. Then, the 
SeqMinCut algorithm (Section~VI) indirectly captures inter-RN interference through the edge weight matrix 
constructed from inter-RN RSRP, thereby reflecting SINR effects during cluster formation. }

\section{Clustered Cells and \ac{DAS}}
\label{sec:Cluster}
Once the initial RN positions are determined, the next task is to cluster multiple RNs to form larger cells~\cite{gesbert2010multicell,irmer2011comp,3gpp36819comp}. This is motivated by two facts: i) in typical enterprise 5G networks several cells experience long periods of under-load or idle operation, and ii) combining multiple RNs improves the downlink \ac{RSRP} of the UEs and reduces inter-cell interference, thereby enhancing the downlink \ac{SINR}. However, in order to maintain orthogonal pilots among intra-cell UEs, admission control is imposed on a per-cell basis~\cite{marzetta2010noncooperative}. This limits the number of RNs that cluster together to form larger cell. To balance this trade-off, we propose the following approach.

\subsection{Sequential \texttt{mincut}}
In this approach (Algorithm~\ref{alg:seqmincut}), we first initialize a completely connected graph ${\bf G}$ where \( V(G) \) denotes the set of vertices comprising of the the RNs in the network. The locations of the vertices obtained via one of the algorithms of the previous section. The edges of the graph ${\bf R}$ are the inter-RN \ac{RSRP} measurements. Thus ${\bf R}$ is a symmetric matrix with the diagonal elements set as $\infty$. Furthermore, ${\bf w}$ is a vector of length $N$ that contains the UE association to each RN. Finally, we consider a system parameter $u_{\max}$ that is the maximum number of UEs admissible in a single cell.

\begin{algorithm}[h]
\caption{$[C]$ = \texttt{SeqMinCut}$(N, {\bf w}, {\bf R}, u_{\max})$}
\begin{algorithmic}[1]
\Require: \begin{itemize}
    \item $N$ vertices representing RNs.
    \item Vertex weights ${\bf w}$: number of UEs with each RN.
    \item Edge weight matrix ${\bf R}$ denoting the inter-RN RSRP. 
    \item threshold $u_{\max}$.
\end{itemize}
\Ensure Clusters of radio nodes such that the sum of vertex weights in each cluster is less than $u_{\max}$.
\State Initialize a fully connected graph with index vector ${\bf g}$ with $N$ nodes.
\State Initialize sub-graph ${\bf C} = \left[{\bf g}\right]$.
\State Initialize $Flag = 0$.
\State Initialize ${\bf T} = []$.
\While{$Flag =0$}
    \For{i = 1:rows({\bf C})}
    \If {$\sum({\bf w}({\bf C}(i,:)) > u_{\max}$}
        \State $[{\bf s}_1, {\bf s}_2]$ = Stoer-Wagner Mincut of $C(i,:)$.
        \State ${\bf T} = [{\bf T}^T,  {\bf s_1}^T, {\bf s_2}^T]^T$
    \Else
        \State ${\bf T} = [{\bf T}^T,  {\bf C(i,:)}^T]^T$
    \EndIf
    \EndFor
    \If {Sum of weights of all sub-graph rows $\leq$ $u_{\max}$}
        \State Flag = 1 and Break
    \Else
        \State Clear $C$;  ${\bf C} = {\bf T}$; ${\bf T} =[]$
    \EndIf
\EndWhile
\State \textbf{return} $C$
\end{algorithmic}
\label{alg:seqmincut}
\end{algorithm}
The minimum cut in ${\bf G}$ is found using a the classical Stoer-Wagner algorithm~\cite{stoer1997simple}, which separates the graph into two disjoint subsets to minimize the sum of weights of the edges crossing the cut. (For large graphs, randomized/global min-cut variants can be considered as alternatives~\cite{karger1993global}.) The nodes of ${\bf G}$ are partitioned based on the minimum cut, resulting in connected components ${\bf s}_1$ and ${\bf s}_2$. The subgraph matrix ${\bf C}$ is updated by replacing $G$ with two rows ${\bf s}_1$ and ${\bf s}_2$. Then for each row of ${\bf C}$, the sum of vertex weights is computed. The \texttt{mincut} is repeated for all those rows whose sum of vertex weights exceed $u_{\max}$. If the sum is less than $u_{\max}$ for all rows of ${\bf C}$, the algorithm terminates and returns the clusters of radio nodes, each cluster containing nodes whose sum of vertex weights is less than $u_{\max}$.

 { The sequential mincut algorithm operates on the RN-UE bipartite graph $G = (V, E)$, where $V$ includes both RNs and UEs and the edge weights reflect link strengths. By recursively applying a minimum $s$-$t$ cut, the algorithm partitions the graph to minimize the number (and weight) of inter-cluster edges. This corresponds to minimizing the aggregate inter-cluster interference, which is directly relevant for coordinated transmission. In contrast, commonly used methods such as $k$-means cluster RNs based on Euclidean distance, which may not reflect interference relationships, and hierarchical methods have higher computational complexity.}

Indeed other approaches for forming optimal \acp{DAS} can be explored in a future work. However, in the next section, with numerical simulations we show the efficacy of \texttt{SeqMinCut} in enhancing the \ac{UE} signal coverage. {From an analytical standpoint, however, Table~\ref{tab:clustering_comparison} compares the sequential mincut approach with $k$-means and hierarchical clustering in terms of objective, complexity, and performance on a representative scenario. As shown, mincut achieves lower inter-cluster interference with comparable 
or better computational efficiency. This motivates its use in the proposed framework.}

\begin{table*}[t]
\centering
\begin{tabular}{|p{3cm}|p{4.5cm}|p{2.4cm}|p{4cm}|}
\hline
\textbf{Method} & \textbf{Clustering Objective} & \textbf{Complexity} & \textbf{Inter-cluster Interference} \\ \hline
\textbf{Sequential Mincut (proposed)} 
& Minimizes the weight of inter-cluster edges in the RN-UE bipartite graph, thus directly targets interference minimization 
& $\mathcal{O}(K \cdot (|V|+|E|))$ 
& Lowest interference among methods \\ \hline
\textbf{$k$-means} 
& Minimizes Euclidean distance between cluster centroids and points 
& $\mathcal{O}(n \cdot K \cdot I)$ where $I$ is iterations 
& Higher interference due to geometric clustering ignoring link structure (benchmark: $-3$ dB) \\ \hline
\textbf{Hierarchical} 
& Agglomerative or divisive linkage based on distance metrics 
& $\mathcal{O}(n^3)$ (agglomerative) 
& Moderate interference reduction but high computational cost \\ \hline
\end{tabular}
\caption{Comparison of clustering algorithms for RN grouping}
\label{tab:clustering_comparison}
\end{table*}




    
    




\section{Numerical Results and Discussion}
\label{sec:NRD}
Here we discuss some numerical results to highlight our key findings. The propagation environment is based on 3GPP specifications for indoor scenarios~\cite{3GPPmodel,3gpp38901,winner22007}. The path-loss model consists of a single slope (see \cite{haneda2016indoor} for details) with path-loss exponent $\alpha_{\rm r} = 3$. Since outdoor MBSs are typically in non-line-of-sight, $\alpha_{\rm m}$ is assumed to be 4~\cite{ghatak2018coverage}. Additionally, a 30 dB outdoor-to-indoor penetration loss is considered for the MBS tier~\cite{sheikh2021outdoor}. The transmit powers of indoor RNs and MBSs are assumed to be 23 dBm and 40 dBm, respectively. The carrier frequency is 3.2 GHz and the total bandwidth per cell is 100 MHz. The noise power density is assumed to be -176 dBm/Hz and a noise figure of 8 dB is considered. The RN mounting height is assumed to be 4 m. Furthermore, we consider that MBSs that are within 1 km from the deployment area contribute to the interference observed by the \ac{UE}.

{In the SG-based analysis, the RN locations are generated according to the BPP model and are independent of the UE locations. All performance metrics are averaged over multiple independent realizations of both RN and UE point processes. In contrast, for the ILP-based optimization 
and clustering methods, RN locations are first optimized or clustered based on a representative UE distribution and then fixed. The network performance is subsequently evaluated by averaging over multiple independent UE realizations drawn from the same distribution. This decoupled procedure ensures that 
the numerical results reflect average network performance under the considered deployment strategy.}

{{
\subsection{A Note on Our Semi-Analytical Simulation Framework} The derivation of the complex moments under MBS cooperation involves conditional Palm probabilities associated with the non-stationary point process and the cooperative association model. These conditional  Palm probabilities characterize the distribution of interferers and serving MBSs conditioned on the presence  of multiple cooperating transmitters. While obtaining their closed-form expressions is in general mathematically intractable, the presented expressions provide a precise analytical framework for extending  meta-distribution analysis to cooperative, non-stationary settings. To aid their evaluation, we adopt a hybrid semi-analytical strategy that alleviates the need to calculate the complex moments. This involves first evaluating the conditional success probabilities analytically (which is feasible since the corresponding expressions are in closed form). Then, this closed form conditional success probability is evaluated for a realization of the point process using simulations. In essence, each Monte Carlo evaluation results in the generation of the random locations, and then the closed form equation provides one value of the conditional success probability, which is then averaged out across realizations.}}

\subsection{Insights from MD analysis:} 
\begin{figure}
    \centering
    \subfloat[]{\includegraphics[width = 0.45\linewidth]{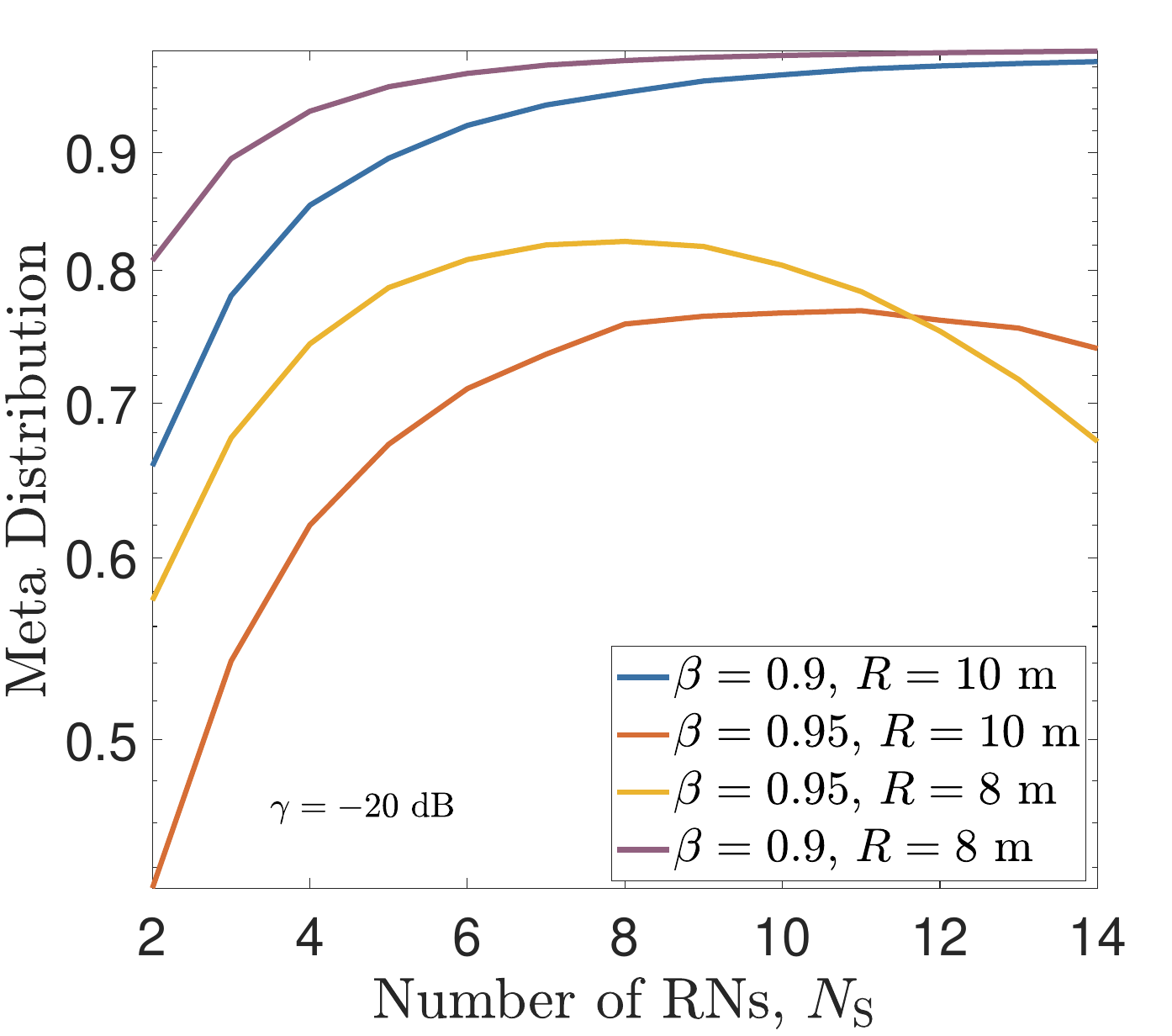}  \label{fig:MDvRN}}
    \hfil
    \subfloat[]{\includegraphics[width = 0.45\linewidth]{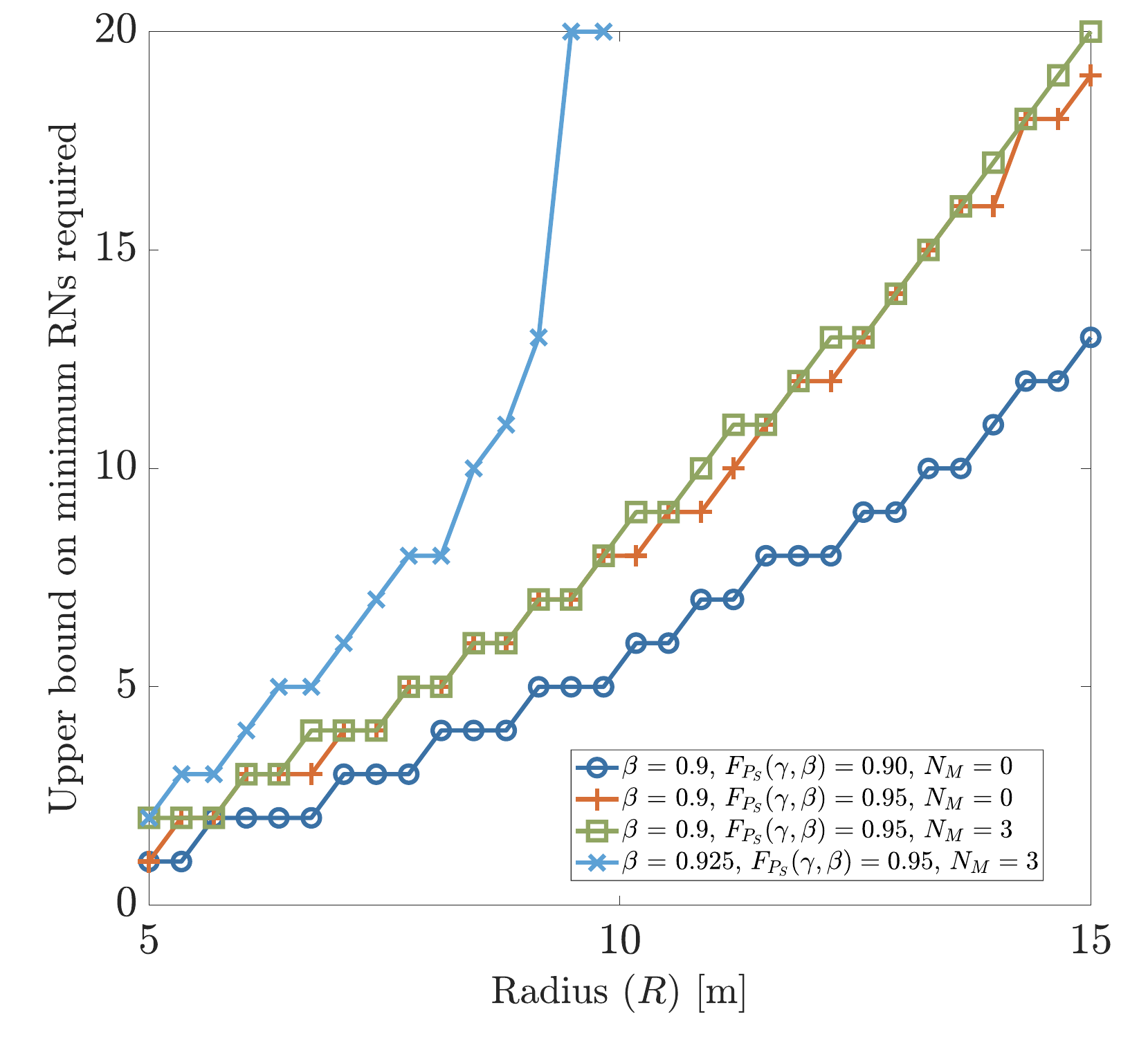} \label{fig:RNvR}}
    \caption{(a) \ac{MD} of the network with respect to the number of RNs; (b) Upper bound on the number of RNs.}
\end{figure}
\begin{figure}
    \centering
    \includegraphics[width = 0.68\linewidth]{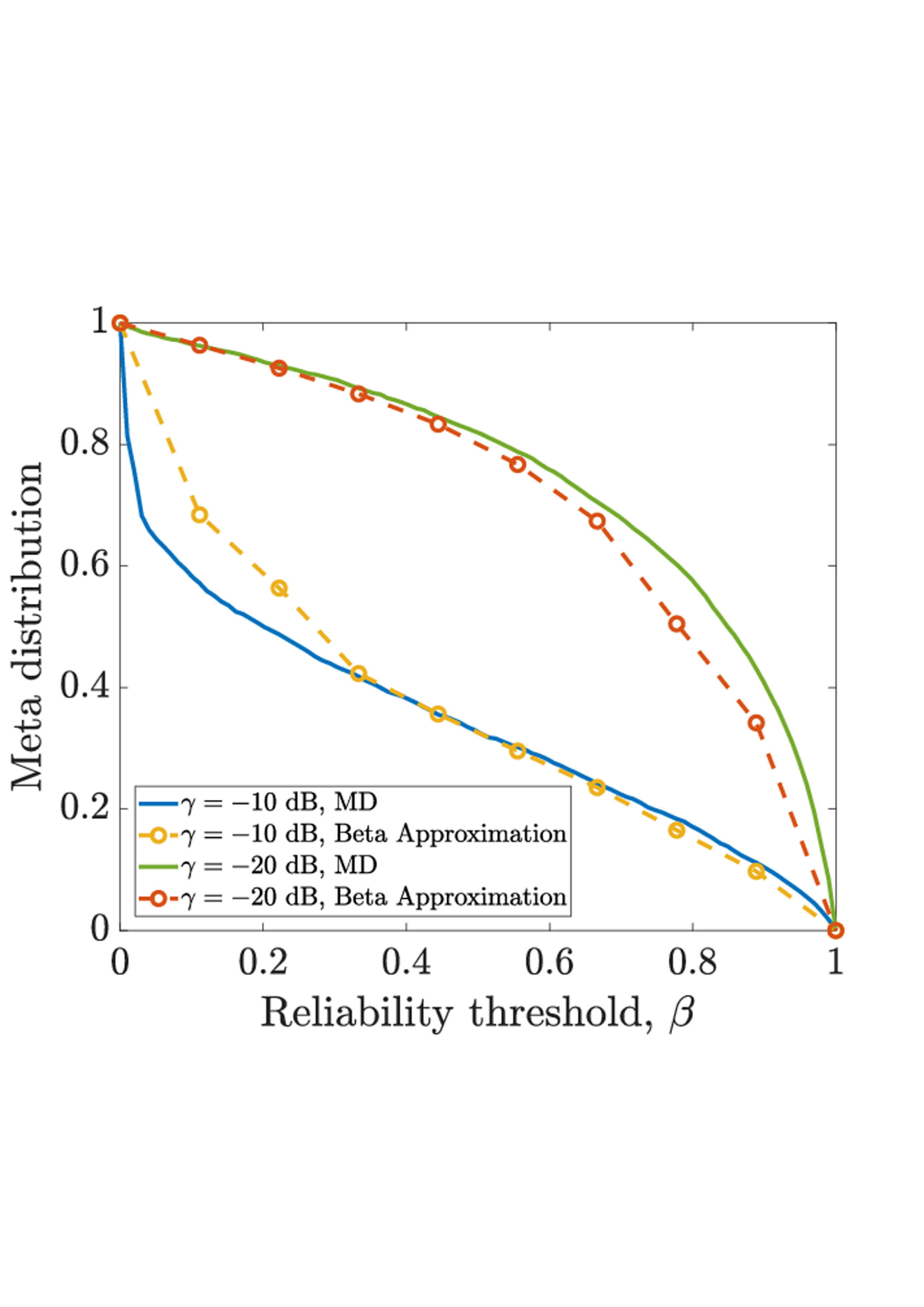}
    \caption{Beta approximation of the MD.}
    \label{fig:md_basic}
\end{figure}
Fig.~\ref{fig:MDvRN} shows the fraction of UEs achieving downlink SINR above $-20$ dB for at least a fraction $\beta$ of time, for different deployment radii, while Fig.~\ref{fig:RNvR} shows the minimum number of \acp{RN} required to meet these coverage targets. As the number of \acp{RN} increases, coverage initially improves, may peak, and can eventually degrade, though this trend is not visible for $\beta=0.9$ in the plotted range. The results also reveal infeasible scenarios; for example, with $R=10$ m, achieving 90\% UE coverage at $\beta=0.95$ is impossible without clustering or optimized placement. Moreover, Fig.~\ref{fig:RNvR} indicates that $\beta$ has a stronger influence on the required number of \acp{RN} than the target UE coverage fraction.
\begin{figure}
    \centering
    \includegraphics[width = 0.65\linewidth]{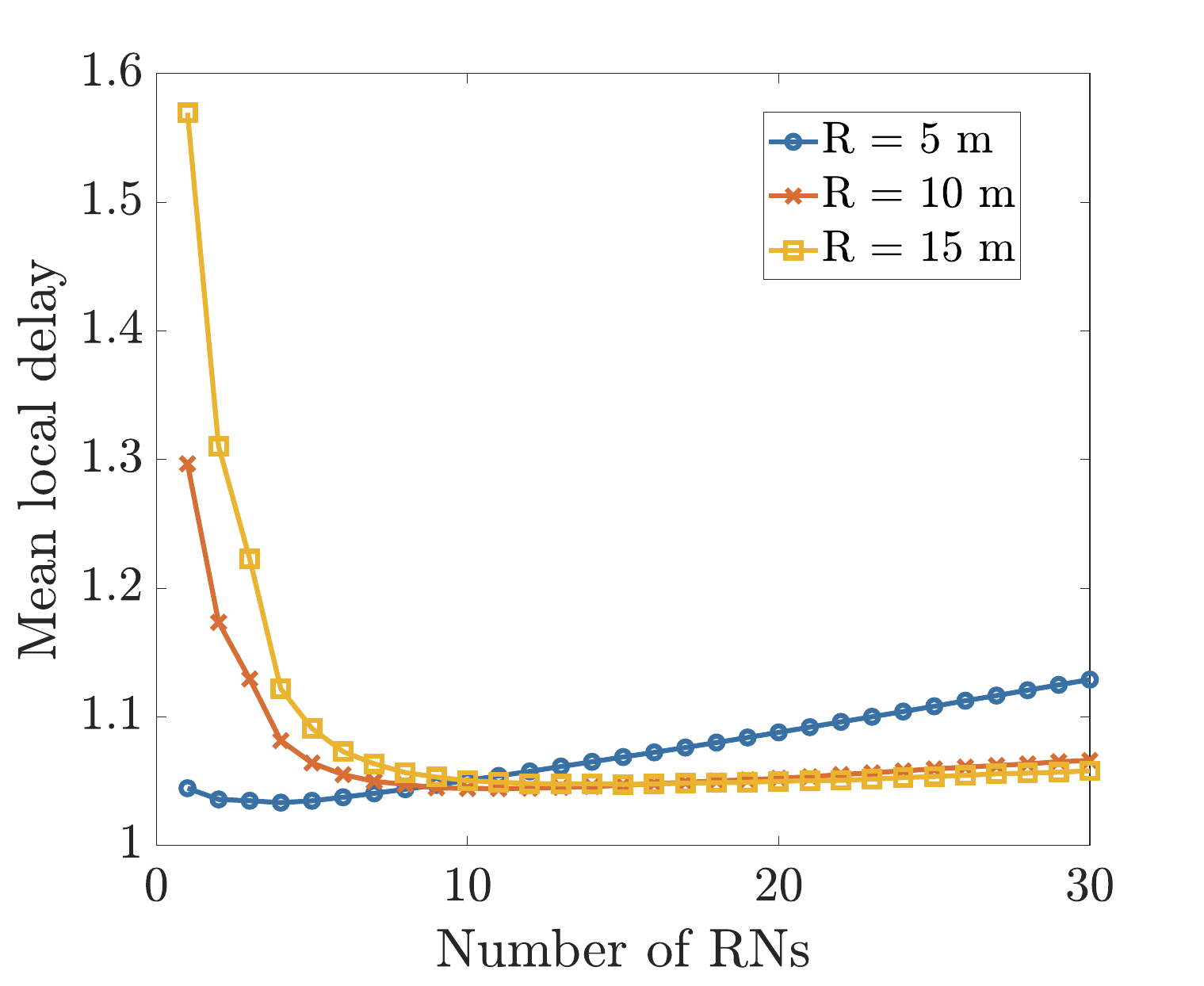}
    \caption{Mean local delay vs the number of RNs.}
    \label{fig:MLD}
\end{figure}

Fig.~\ref{fig:MLD} shows the mean local delay for different number of RNs and different deployment sizes. Recall that the mean local delay is defined as the number of attempts required by the test UE for a single transmission success. For a given network area, there exists an optimal number of RNs that minimizes the mean local delay. This optimal RN number may be different that the one obtained for a given SINR threshold and reliability threshold using the MD framework. Thus, based on the objective to be optimized, the operator needs to select an appropriate RN number to plan the network. For smaller networks, although increasing the number of RNs may increase the RSRP, we observe that the mean local delay sharply increases due to an increase in SINR.

{
\subsection{Sensitivity of RN dimensioning to modelling assumptions:}
The disk-shaped region, uniform UE law, and homogeneous propagation model are analytical abstractions. To quantify the sensitivity of the resulting dimensioning rule, we keep the service-level target fixed and perturb one modelling assumption at a time. Specifically, we determine
\[
N^{\star}=\min\left\{N:{F}_{P_s}\!\left(\beta,\gamma;N\right)\geq\eta\right\},
\]
with $\beta=0.90$ and $\eta=0.95$, while retaining the SINR threshold $\gamma$ and the remaining parameters of the baseline experiment. The cases considered are: (i) a disk, an equal-area $2{:}1$ rectangle, and an equal-area irregular footprint; (ii) uniform, two-hotspot, and edge-biased UE laws with the same mean UE density; and (iii) variations of the path-loss exponent and a moderate blockage/shadowing setting. All cases use the same Monte-Carlo budget and common random numbers.
}
{
\begin{table*}
  \begin{center}
\renewcommand{\arraystretch}{1.12}
\begin{tabular}{p{0.22\linewidth}p{0.42\linewidth}ccc}
\toprule
Factor & Setting & $N^{\star}$ & $\Delta N$ & Relative change \\
\midrule
Baseline & Disk, uniform UEs, $\alpha=3$ & $8$ & $0$ & $0\%$ \\
Footprint & Equal-area $2{:}1$ rectangle & $9$ & $+1$ & $+12.5\%$ \\
Footprint & Equal-area irregular region & $10$ & $+2$ & $+25.0\%$ \\
UE law & Two-hotspot mixture & $11$ & $+3$ & $+37.5\%$ \\
UE law & Edge-biased distribution & $10$ & $+2$ & $+25.0\%$ \\
Propagation & $\alpha=2.5$ & $7$ & $-1$ & $-12.5\%$ \\
Propagation & $\alpha=3.5$ & $10$ & $+2$ & $+25.0\%$ \\
Propagation & Moderate blockage/shadowing & $11$ & $+3$ & $+37.5\%$ \\
\bottomrule
\end{tabular}
\end{center}
\caption{Sensitivity of the MD based dimensioning to system parameters.}
\label{tab:sensitivity}
\end{table*}

The results (see Tab.~\ref{tab:sensitivity}) show that the RN recommendation is comparatively insensitive to moderate changes in footprint shape, but is more strongly affected by a mismatch between the uniform RN model and a concentrated UE distribution, or by harsher propagation. In the considered finite and partially noise-limited regime, the largest change is $37.5\%$. These values are not universal correction factors; rather, they indicate when the SG output can be used as a first-stage estimate and when the geometry-specific ILP and ray-tracing stages are necessary. Accordingly, the SG-derived count is treated as a conservative planning baseline rather than as a deployment-independent optimum.
}

\subsection{MD vs ILP predictions:} Table~\ref{table:comp_table_MD_ILP_result} shows the required number of RNs predicted by the ILP formulation and for different target standard success probabilities $(p_{\rm s})$ of the MD framework. For the ILP solution, we employed a homogeneous 3GPP propagation model and no specific blockage environment was considered. It has long been contested whether SG models such as PPP offer an optimistic~\cite{lee2013stochastic} or a pessimistic~\cite{andrews2011tractable} view of the network. Nevertheless, we see that the 95\% standard success probability requirement closely matches with the ILP predictions. { The required number of RNs, $N$, is determined 
using the $\beta$-percentile of the meta-distribution of the conditional success probability. Specifically, we 
choose $N$ such that 
\(
\mathbb{P}\big(P_{\mathrm{s}}(\theta) \geq p_{\mathrm{th}}\big) \geq \beta,
\)
where $P_{\mathrm{s}}(\theta)$ is the conditional success probability at SIR threshold $\theta$. In this paper, 
we use $\beta = 0.95$ and $p_{\mathrm{th}} = 0.9$ as the design parameters. This ensures that at least 95\% of 
UEs achieve a conditional success probability of 0.9 or higher.}

It must be noted that the ILP framework does not take into account the fading of the individual links and hence can only optimize the placement with respect to averaged channels. On the other hand, the key advantage of an ILP formulation over MD is the possibility of including exact propagation features, such as blockages, walls, etc. Fig.~\ref{fig:ILP_Illustration} shows that for a 30 m $\times$ 30 m network scenario, without any MBS interference, the presence of two walls results in the requirement number of RNs increase from 1 to 4. Thus, in case the network operator has access to the specific indoor blockage features, ILP is a more accurate tool for specific scenarios.

{
\subsection{Numerical illustration of the effect of infrastructure constraints:}
To isolate the effect of physical-deployment constraints, we consider the $48\,\mathrm{m}\times30\,\mathrm{m}$ floorplan shown below. It contains $15$ demand points, nine candidate RN sites, and two PoE closets. Both formulations use the same radio-feasibility matrix: a site can serve a demand point if their horizontal separation is at most $13$~m. The central candidate at $(24,15)$~m lies above a restricted ceiling zone and is therefore assigned $f_s=e_s=0$. The cable-route limit is $L_{\max}=35$~m, and each closet has two available ports. We use normalized costs $c_s^{\mathrm{inst}}=1$ and $c^{\mathrm{cab}}=0.01$ per metre.

\refstepcounter{figure}\label{fig:infrastructure-aware-comparison}
\begin{center}
\includegraphics[width=0.98\linewidth]{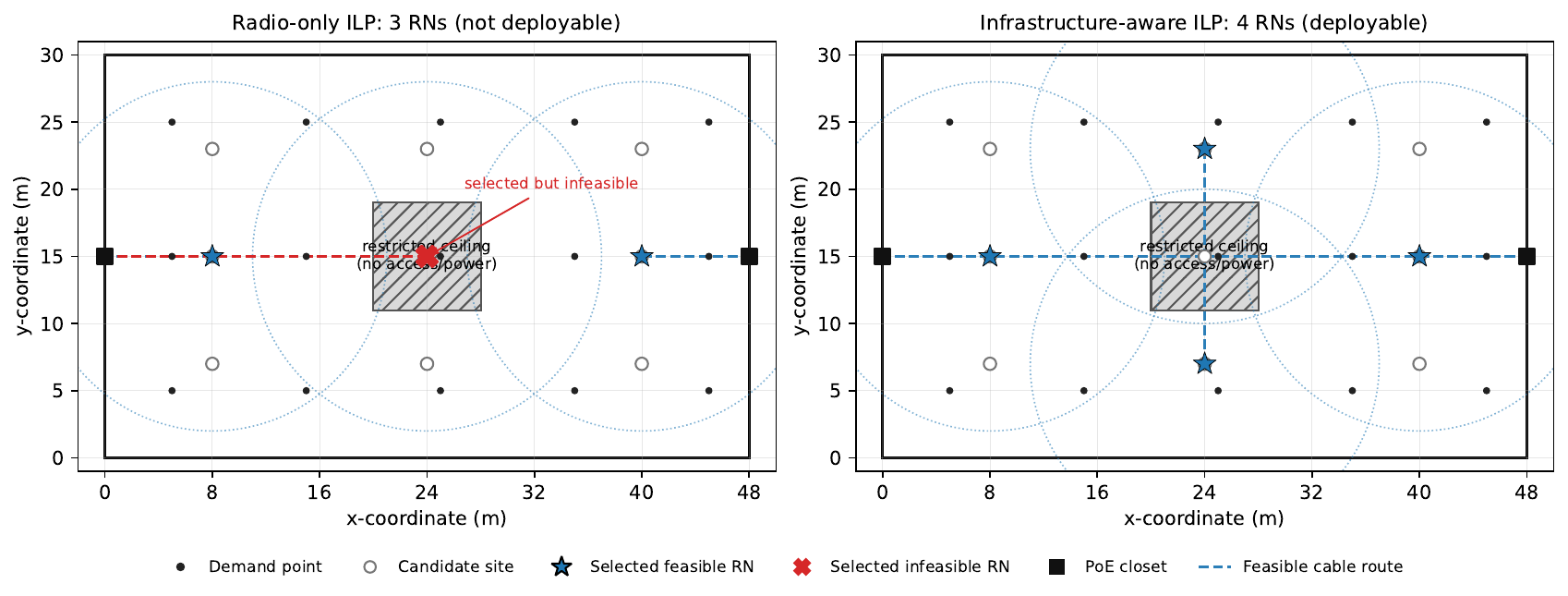}
\end{center}
\noindent\textbf{Fig.~\thefigure.} Comparison of the original radio-only solution and the infrastructure-aware solution. The radio-only ILP selects three centrally located RNs, but one selected site is inaccessible and unpowered. The infrastructure-aware ILP replaces that site by two feasible sites and produces a four-RN placement satisfying the cable-reach and closet-capacity constraints.

\begin{table*}\label{tab:infrastructure-aware-comparison}
\begin{center}
\renewcommand{\arraystretch}{1.12}
\begin{tabular}{lcc}
\toprule
Metric & Radio-only ILP & Infrastructure-aware ILP \\
\midrule
Selected RN locations (m) & $(8,15),(24,15),(40,15)$ & $(8,15),(40,15),(24,7),(24,23)$ \\
Number of RNs & $3$ & $4$ \\
Selected inaccessible/unpowered sites & $1$ & $0$ \\
Total cable length (m) & $40^{\dagger}$ & $80$ \\
Maximum cable-route length (m) & $24^{\dagger}$ & $32$ \\
Mean serving distance (m) & $8.74$ & $7.77$ \\
Maximum serving distance (m) & $12.21$ & $11.18$ \\
Normalized deployment cost & $3.40^{\dagger}$ & $4.80$ \\
Physically deployable & No & Yes \\
\bottomrule
\end{tabular}
\end{center}
\caption{Numerical comparison of the two placement formulations. $^{\dagger}$The cable length and cost for the radio-only solution are post-evaluated nominal values; the placement remains infeasible because one selected site violates the accessibility and power constraints.}
\end{table*}

The original problem therefore returns a lower-cardinality radio solution that cannot be installed. Enforcing the infrastructure constraints increases the RN count from three to four and the normalized cost from $3.40$ to $4.80$, but eliminates all deployment violations. The additional RN also reduces the maximum serving distance from $12.21$~m to $11.18$~m. This example demonstrates that the deployability-aware formulation changes both the feasible set and the optimal placement; the physical constraints cannot, in general, be applied as a harmless post-processing step after radio-only optimization.
}
\begin{table*}[h]
\begin{center}
\renewcommand{\arraystretch}{1.12}
\begin{tabular}{lcccc}
\toprule
Comparison & Mean WKHM gain & $95\%$ CI & $t(29)$ & Holm-adjusted $p$ \\
& (Mbps) & (Mbps) & & \\
\midrule
WKHM-KM  & $2.17$ & $[1.23,\,3.11]$ & $4.72$ & $2.2\times10^{-4}$ \\
WKHM-CKM & $1.58$ & $[0.55,\,2.61]$ & $3.14$ & $1.14\times10^{-2}$ \\
WKHM-KHM & $1.21$ & $[-0.02,\,2.44]$ & $2.01$ & $5.4\times10^{-2}$ \\
WKHM-KC  & $0.23$ & $[0.07,\,0.39]$ & $2.87$ & $1.5\times10^{-2}$ \\
\bottomrule
\end{tabular}
\end{center}
\caption{Statistical significance tests for WKHM.}
\label{tab:sig}
\end{table*}

\begin{figure}
    \centering
    \includegraphics[width = 0.7\linewidth]{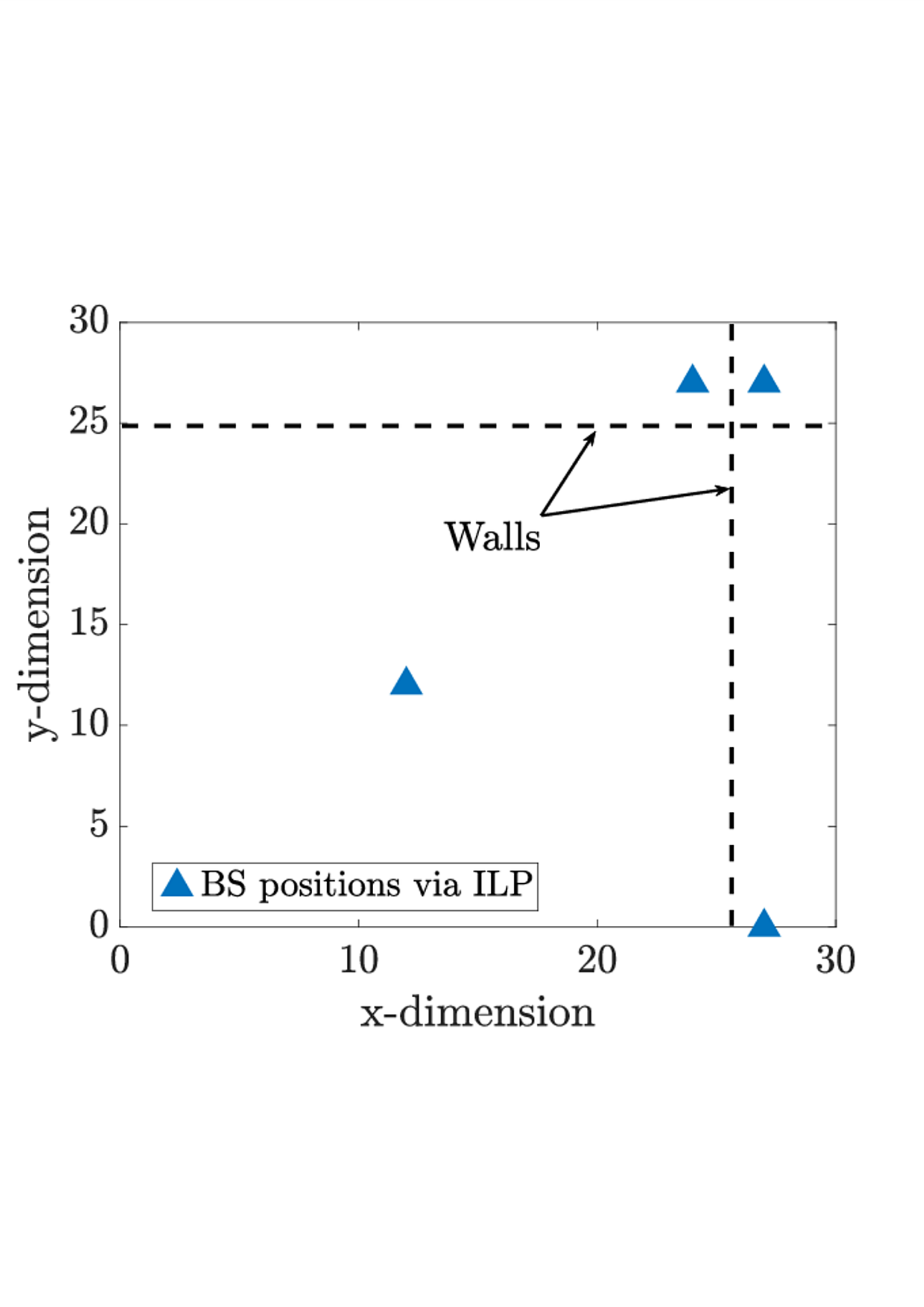}
    \caption{Illustration of ILP taking into account wall blockages.}
    \label{fig:ILP_Illustration}
\end{figure}

\begin{table*}[t]
\centering
\begin{tabular}{|c|c|c|c|c|c|}
\hline
 Square Dimension & Circular Radius &ILP &$p_{\rm s} = 0.9$ &$p_{\rm s} = 0.95$ &Optimal $N_s$ $(p_{\rm s}^*)$ \\ \hline
 10 m $\times$ 10 m & 5.64 m &1 &1 &1 &1(99.7\%) \\ \hline
 30 m $\times$ 30 m & 16.92 m &1 &1 &1 &3(99\%) \\ \hline
 50 m $\times$ 50 m & 28.20 m &1 &1 &1 &7(98.8\%) \\ \hline
 100 m $\times$ 100 m & 56.41 m &1 &1 &2 &11(98.6\%) \\ \hline
 200 m $\times$ 200 m & 112.83 m &1 &1 &3 &24(98.5\%) \\ \hline
 500 m $\times$ 500 m & 282.09 m &5 &2 &5 &29(98.4\%) \\ \hline
 750 m $\times$ 750 m & 423.14 m &7 &3 &6 &36(98.3\%) \\ \hline
 1000 m $\times$ 1000 m & 564.18 m &8 &4 &8 &36(98.3\%) \\ \hline
\end{tabular}
\caption{Required number of RNs predicted by SG analysis and ILP solution.}
\label{table:comp_table_MD_ILP_result}
\end{table*}
\begin{table*}[h]
\centering
\begin{tabular}{|c|c|Sc|p{23mm}|p{25mm}|p{23mm}|p{20mm}|}
\hline
Abscissa&Ordinate &Example&$\mathbb{E}_\Phi[\min_{UE}\{\zeta\}]$&$\min_\Phi\{\min_{UE}\{\zeta\}\}$&$\mathbb{E}_{\Phi}[\mathbb{E}_{UE}[\zeta]]$ & Ratio\\ \hline
Uniform&Uniform &
    \includegraphics[width = 0.1\linewidth]{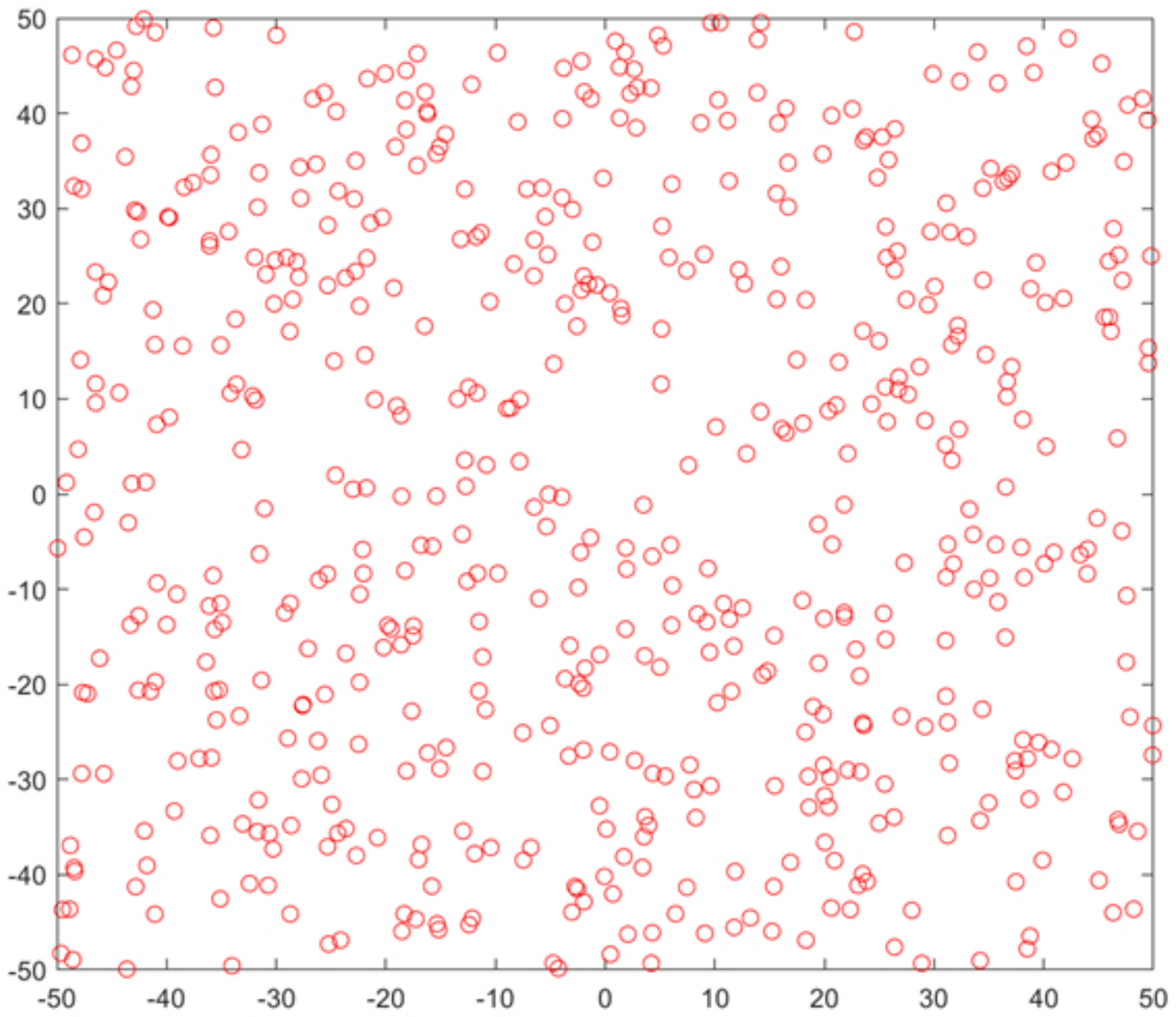}&\vspace{-1.7cm} \textcolor{blue}{KM: -97.13}  \newline CKM: -97.67 \newline KHM: -97.21 \newline \textcolor{Ao}{WKHM: -97.05}\newline \textcolor{red}{KC: -99.92} &\vspace{-1.7cm}\textcolor{blue}{KM: -100.54} \newline CKM: -101.13 \newline KHM: -101.46 \newline \textcolor{Ao}{WKHM: -99.24} \newline \textcolor{red}{KC: -100.91}&\vspace{-1.7cm}KM: -83.64 \newline CKM: -83.65 \newline \textcolor{Ao}{KHM: -83.07}\newline \textcolor{blue}{WKHM: -83.60}\newline \textcolor{red}{KC: -85.30}&\vspace{-1.7cm}KM: 2.08\newline \textcolor{Ao}{CKM: 1.46}\newline KHM: 1.82\newline \textcolor{blue}{WKHM: 1.6}\newline \textcolor{red}{KC: 4.9}\\
\hline
Uniform&Exponential&
    \includegraphics[width = 0.1\linewidth]{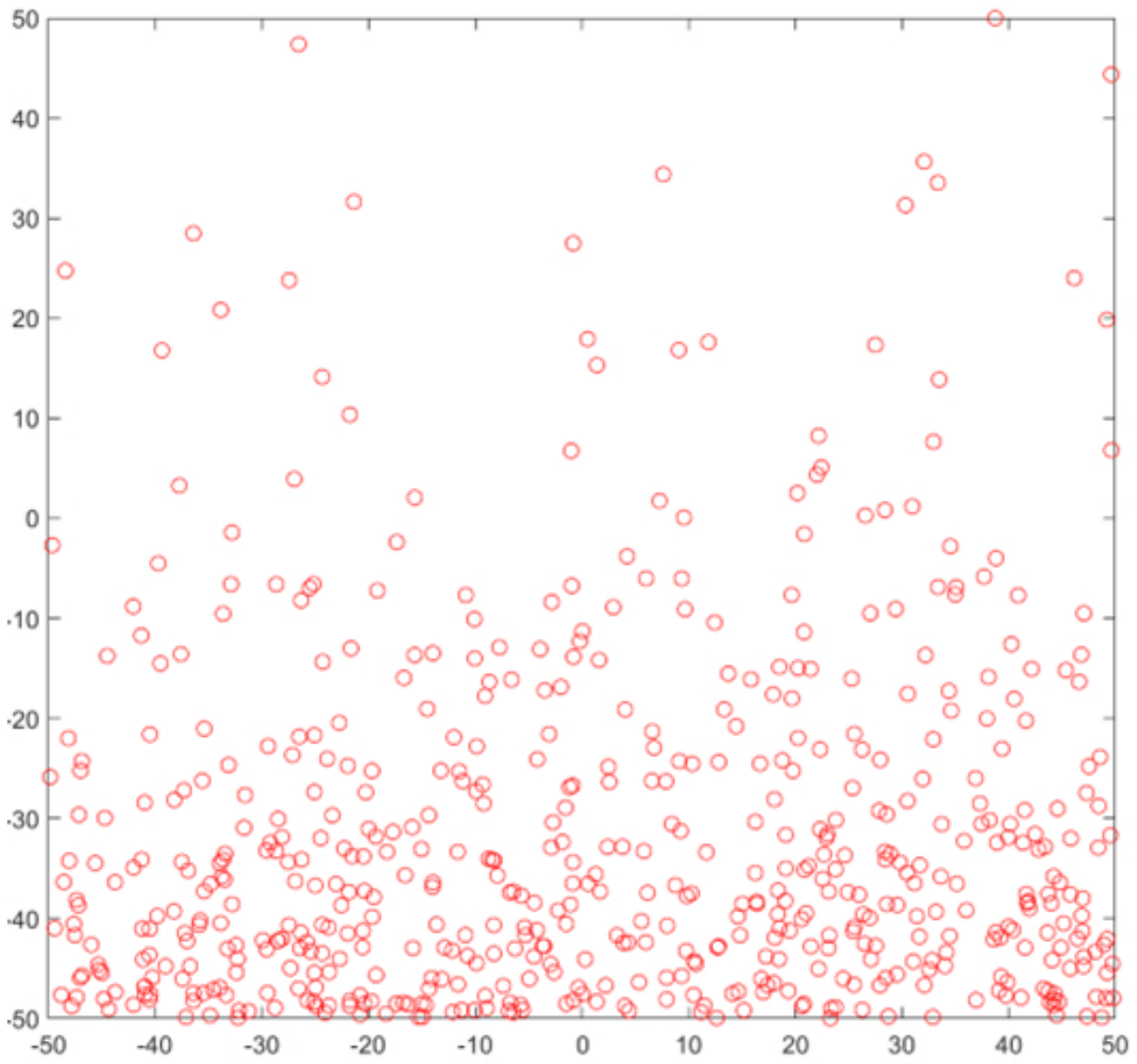}&\vspace{-1.8cm}KM: -104.8 \newline \textcolor{red}{CKM: -109.2} \newline KHM: -107.3 \newline \textcolor{blue}{WKHM: -103.5} \newline \textcolor{Ao}{KC: -99.1}&\vspace{-1.8cm}KM: -107.9\newline \textcolor{red}{CKM: -111.7}\newline KHM: -108.8\newline \textcolor{blue}{WKHM: -105.6}\newline \textcolor{Ao}{KC: -83.4}&\vspace{-1.8cm}KM: -80.6\newline \textcolor{blue}{CKM: -79.9} \newline \textcolor{Ao}{KHM: -79.8}\newline WKHM: -81.1\newline \textcolor{red}{KC: -83.4}&\vspace{-1.8cm}KM: 10.6\newline \textcolor{Ao}{CKM: 1.55}\newline \textcolor{blue}{KHM: 3.7}\newline WKHM: 10.5\newline \textcolor{red}{KC: 71.3}\\
\hline
Uniform&Gaussian&\includegraphics[width = 0.1\linewidth]{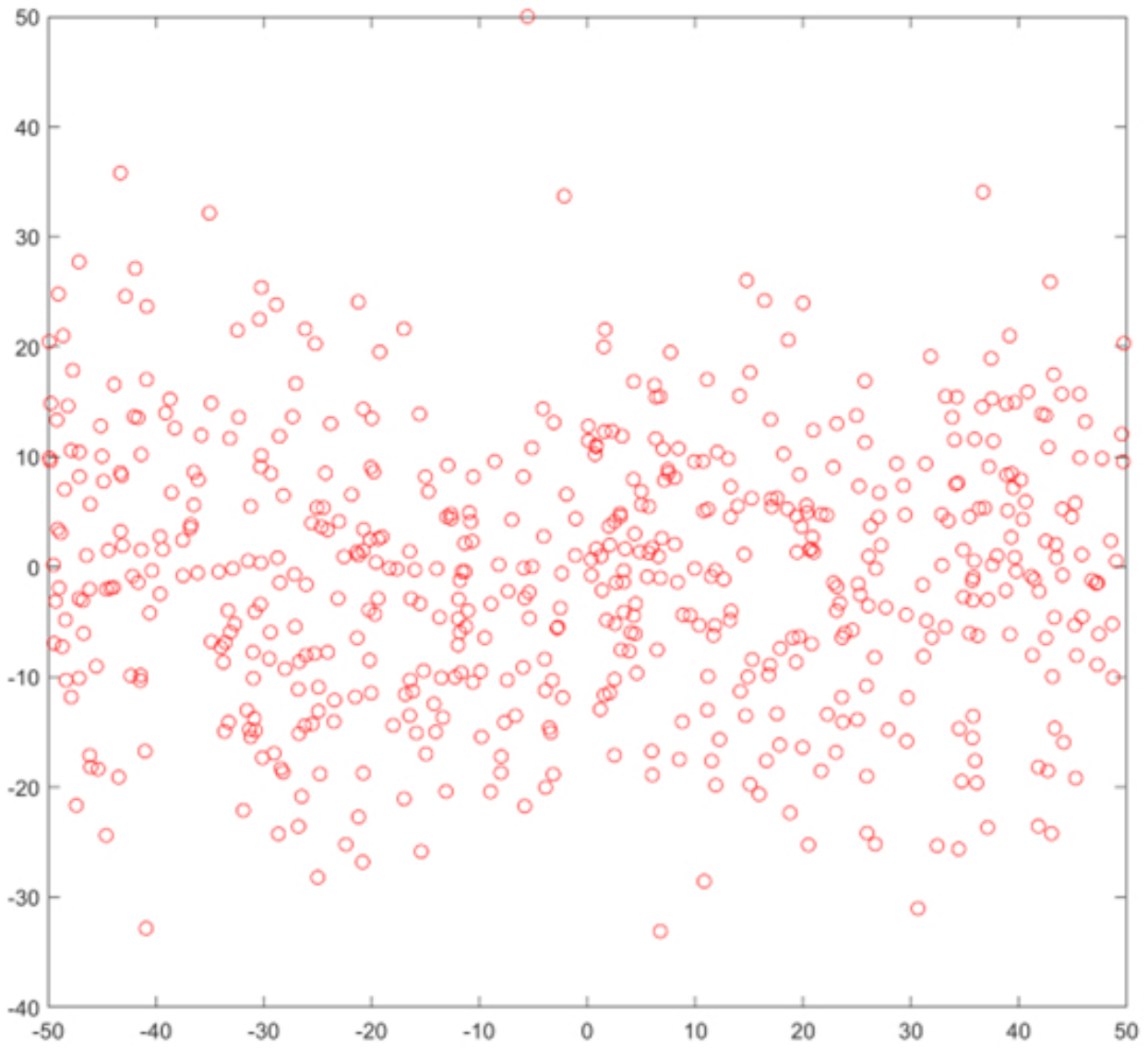}&\vspace{-1.8cm}KM: -100.9 \newline CKM: -101.5\newline \textcolor{red}{KHM: -102}\newline \textcolor{blue}{WKHM: -100.7}\newline \textcolor{Ao}{KC: -98.5}&\vspace{-1.8cm}KM:-102.3 \newline \textcolor{red}{CKM: -103.4}\newline KHM: -103.1\newline \textcolor{blue}{WKHM: -100.5}\newline \textcolor{Ao}{KC: -100.1}&\vspace{-1.8cm}KM: -81.6 \newline CKM: -81.7\newline \textcolor{Ao}{KHM: -81.1}\newline \textcolor{blue}{WKHM: -81.5}\newline \textcolor{red}{KC: -84.5}&\vspace{-1.8cm}KM: 2.6\newline \textcolor{blue}{CKM: 1.6}\newline KHM: 1.7\newline \textcolor{Ao}{WKHM: 1.5}\newline \textcolor{red}{KC: 22}\\ \hline
Uniform&Bi-exponential&\includegraphics[width = 0.1\linewidth]{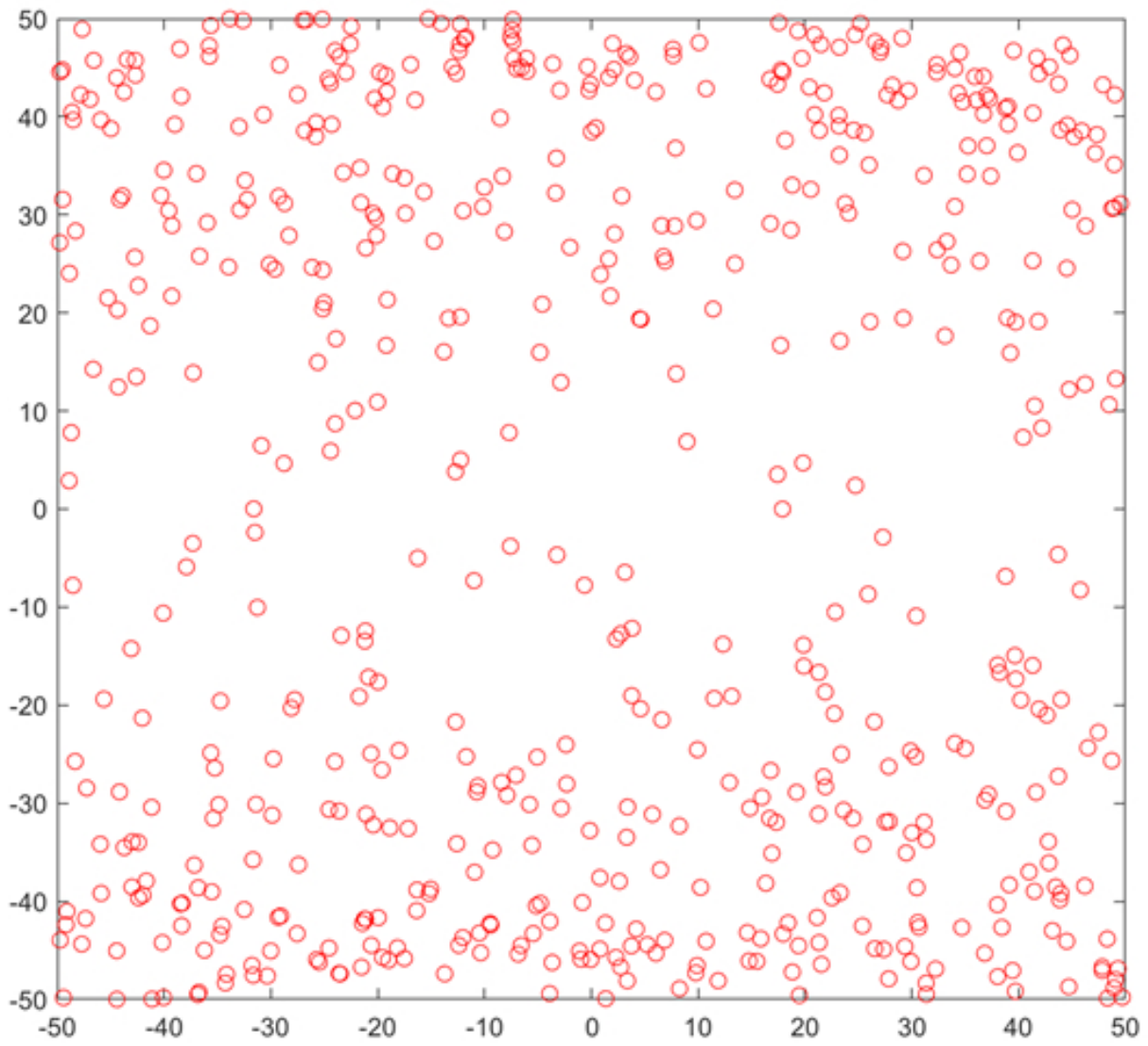}&\vspace{-1.8cm}\textcolor{blue}{KM: -98.8}\newline CKM: -99.2\newline KHM: -99.4\newline \textcolor{Ao}{WKHM: -98.3}\newline \textcolor{red}{KC: -100.2}&\vspace{-1.8cm}KM: -101.6\newline CKM: -101.6\newline \textcolor{red}{KHM: -102.7}\newline \textcolor{Ao}{WKHM: -100.3}\newline \textcolor{blue}{KC: -101.4}&\vspace{-1.8cm}KM: \textcolor{blue}{-83.3}\newline CKM: -83.4\newline \textcolor{Ao}{KHM: -83}\newline WKHM: -83.6 \newline \textcolor{red}{KC: -84.4}&\vspace{-1.8cm}KM: 2.4\newline \textcolor{Ao}{CKM: 1.5}\newline \textcolor{blue}{KHM: 2.1}\newline WKHM: 2.6 \newline \textcolor{red}{KC: 3.7}\\ \hline
Exponential&Exponential&\includegraphics[width = 0.1\linewidth]{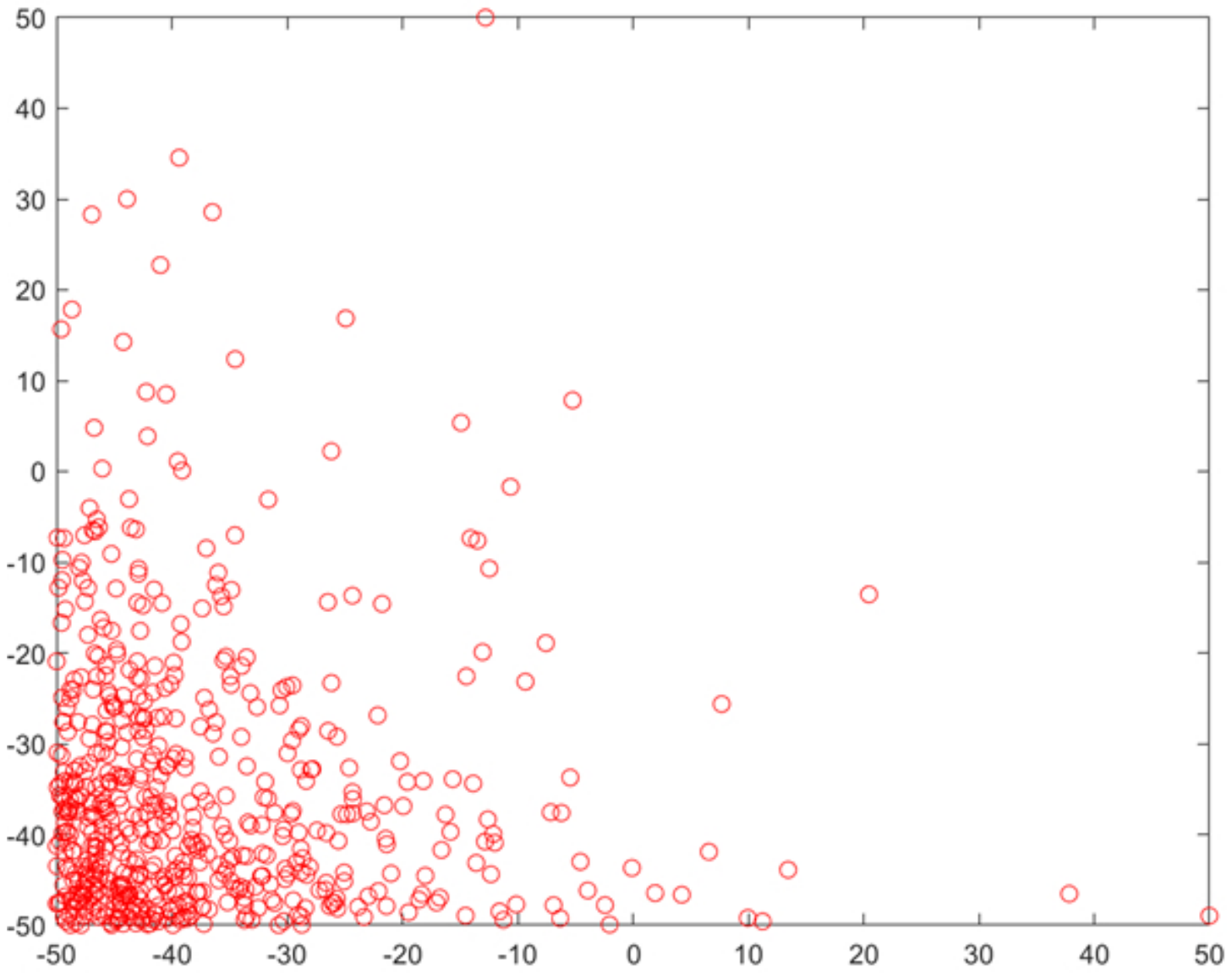}&\vspace{-1.6cm}KM: -102.8\newline \textcolor{red}{CKM: -106.2}\newline KHM: -104.9\newline \textcolor{blue}{WKHM: -99.3}\newline \textcolor{Ao}{KC: -95.6}&\vspace{-1.6cm}KM: -106.8\newline \textcolor{red}{CKM: -108.8}\newline KHM: -106.9\newline \textcolor{blue}{WKHM: -100}\newline \textcolor{Ao}{KC: -97.3}&\vspace{-1.6cm}KM: -78.8\newline \textcolor{Ao}{CKM: -78.0}\newline \textcolor{blue}{KHM: -78.3}\newline WKHM: -79.3\newline \textcolor{red}{KC: -82.0}& \vspace{-1.6cm}KM: 17\newline \textcolor{Ao}{CKM: 1.7}\newline \textcolor{blue}{KHM: 8.7}\newline WKHM: 23.5\newline \textcolor{red}{KC: 170}\\ \hline
Exponential&Gaussian&\includegraphics[width = 0.1\linewidth]{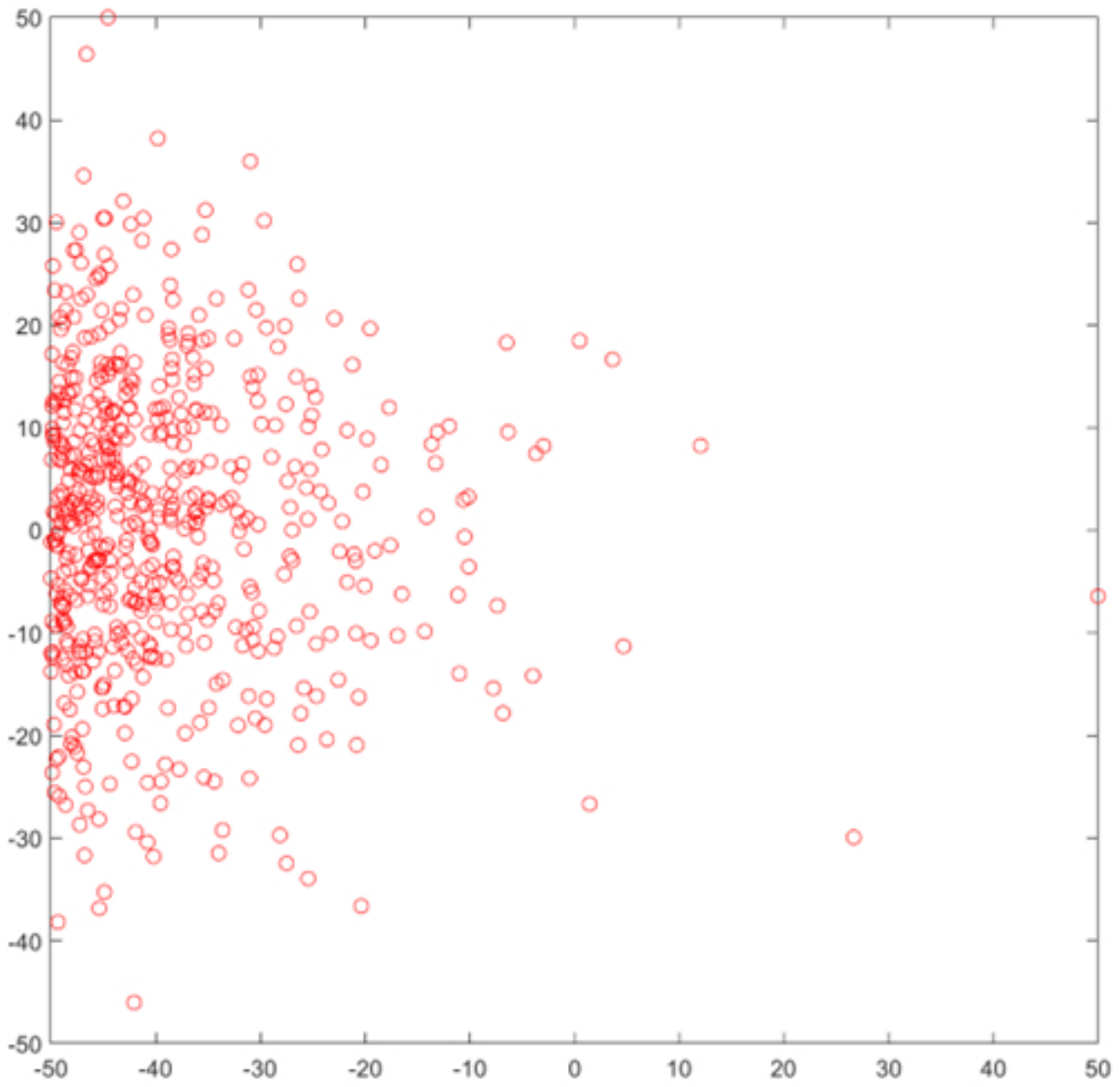}&\vspace{-1.8cm}KM: -101.1\newline \textcolor{red}{CKM: -107.6}\newline KHM: -102.5\newline \textcolor{blue}{WKHM: -99.7}\newline \textcolor{Ao}{KC: -98.1}&\vspace{-1.8cm}KM: -106.9\newline \textcolor{red}{CKM: -108.2}\newline KHM: -107.7\newline \textcolor{blue}{WKHM: -103.0}\newline \textcolor{Ao}{KC: -98.9}&\vspace{-1.8cm}KM: -78.3 \newline \textcolor{Ao}{CKM: -77.5}\newline \textcolor{blue}{KHM: -77.7}\newline WKHM: -78.4\newline \textcolor{red}{KC: -81.3}&\vspace{-1.8cm}KM: 16.4\newline \textcolor{Ao}{CKM: 1.6} \newline \textcolor{blue}{KHM: 9.3} \newline WKHM: 16.2\newline \textcolor{red}{KC: 301}\\ \hline
Exponential&Bi-exponential&\includegraphics[width = 0.1\linewidth]{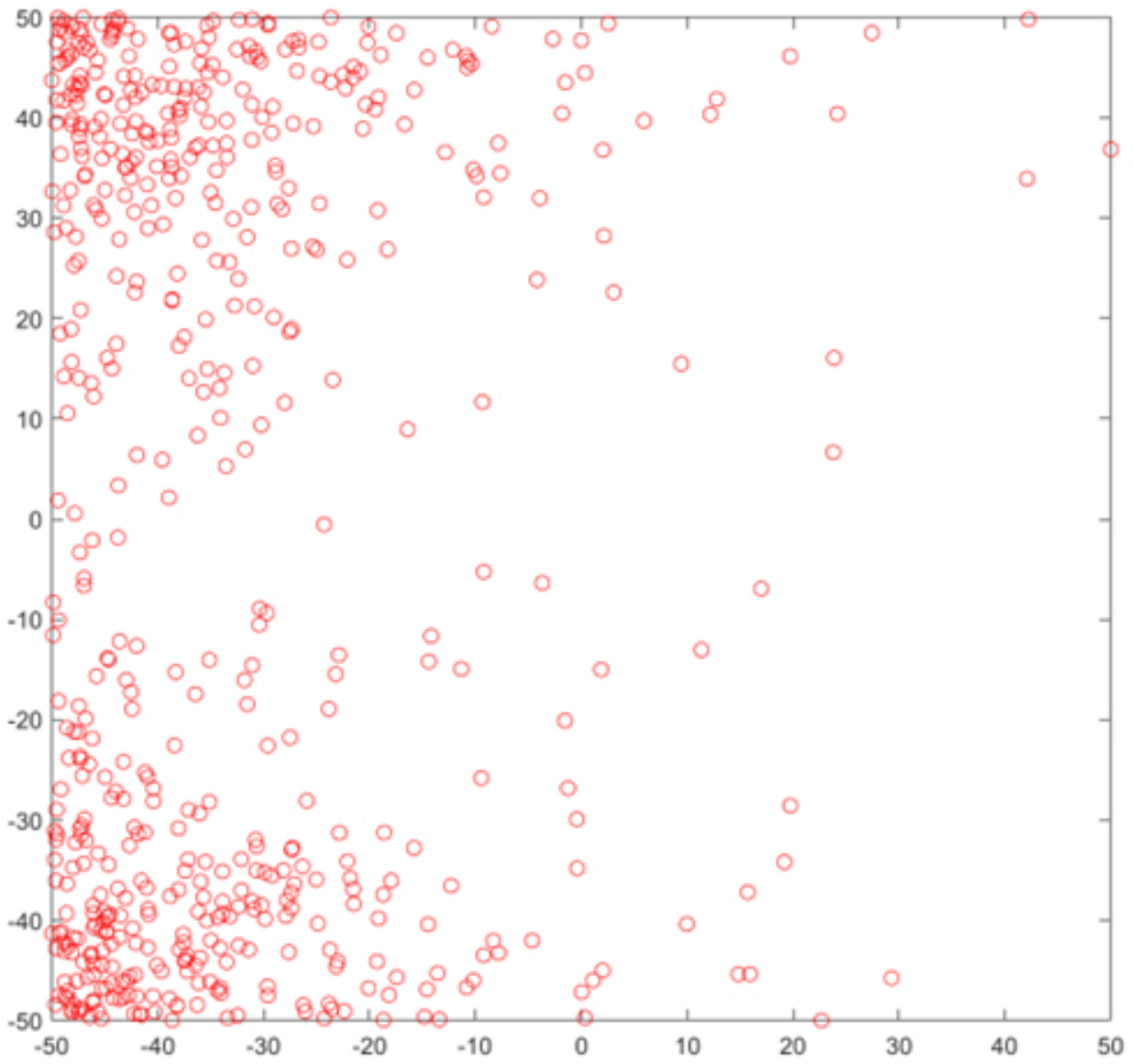}&\vspace{-1.8cm}KM: -105.4\newline \textcolor{red}{CKM: -107.8}\newline KHM: -107.5\newline \textcolor{blue}{WKHM: -103.4}\newline \textcolor{Ao}{KC: -99.1}&\vspace{-1.8cm}KM: -108.8\newline CKM: -108.8\newline \textcolor{red}{KHM: -109.5}\newline \textcolor{blue}{WKHM: -107.0}\newline \textcolor{Ao}{KC: -100.1}&\vspace{-1.8cm}KM: -80.5\newline \textcolor{Ao}{CKM: -79.8}\newline \textcolor{blue}{KHM: -79.9}\newline WKHM: -80.7\newline \textcolor{red}{KC: -82.3}&\vspace{-1.8cm}KM: 17.8\newline \textcolor{Ao}{CKM: 1.88}\newline \textcolor{blue}{KHM: 6.6}\newline WKHM: 15.6 \newline \textcolor{red}{KC: 209}\\ \hline
Gaussian&Gaussian&\includegraphics[width = 0.1\linewidth]{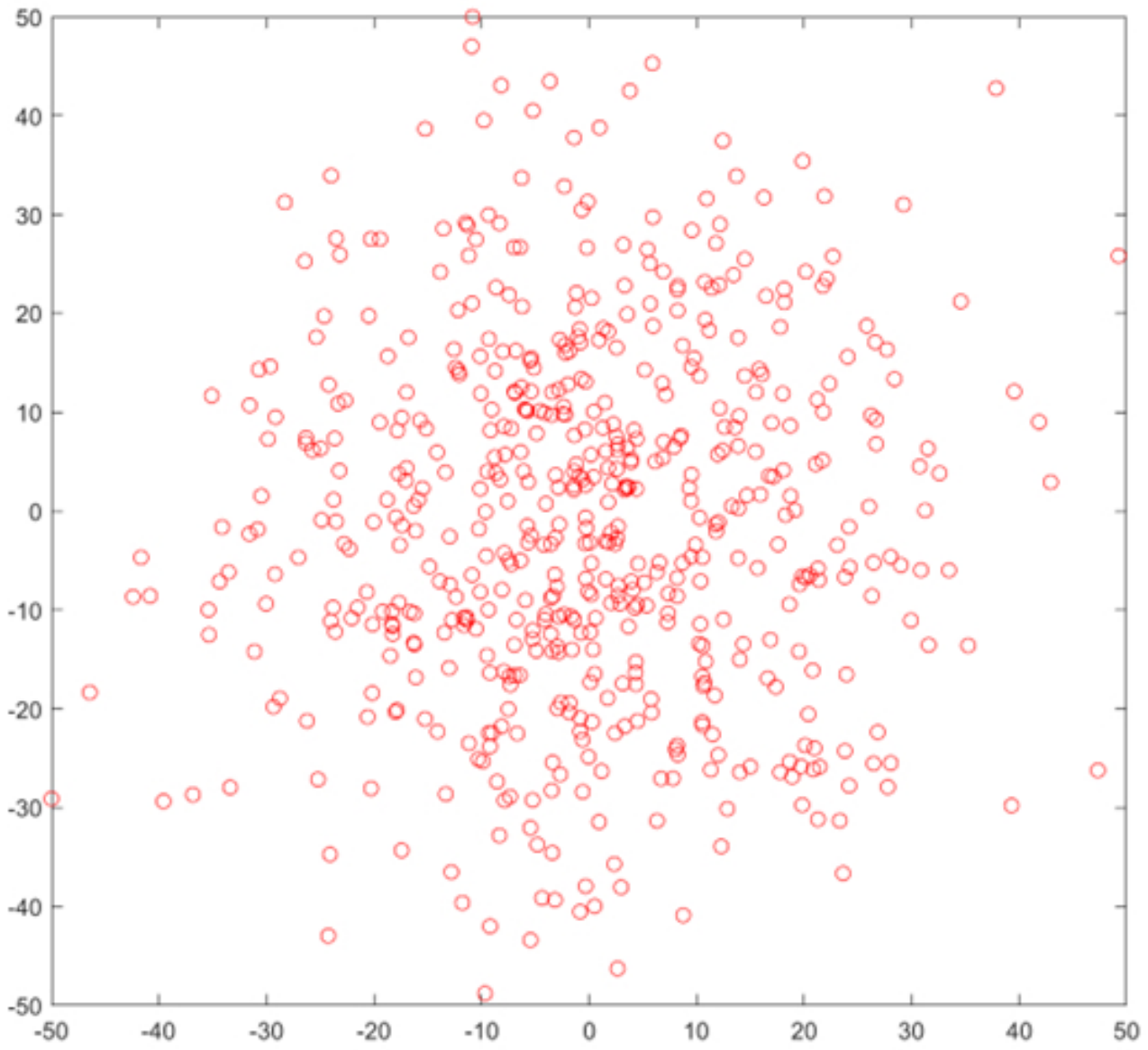}&\vspace{-1.8cm}KM: -99.3\newline CKM: -99.5\newline \textcolor{red}{KHM: -100.9}\newline \textcolor{blue}{WKHM: -99.2}\newline \textcolor{Ao}{KC: -97.8}&\vspace{-1.8cm}KM: -100.9\newline \textcolor{red}{CKM: -102.2}\newline KHM: -101.4\newline \textcolor{blue}{WKHM: -100.4}\newline \textcolor{Ao}{KC: -98.7}&\vspace{-1.8cm}KM: -81.1\newline CKM: -80.9\newline \textcolor{Ao}{KHM: -80.2}\newline \textcolor{blue}{WKHM: -80.8}\newline \textcolor{red}{KC: -83.9}&\vspace{-1.8cm}KM: 4.2\newline \textcolor{Ao}{CKM: 1.7} \newline \textcolor{blue}{KHM: 2.1}\newline WKHM: 3.11\newline \textcolor{red}{KC: 40}\\ \hline
Gaussian&Bi-exponential&\includegraphics[width = 0.1\linewidth]{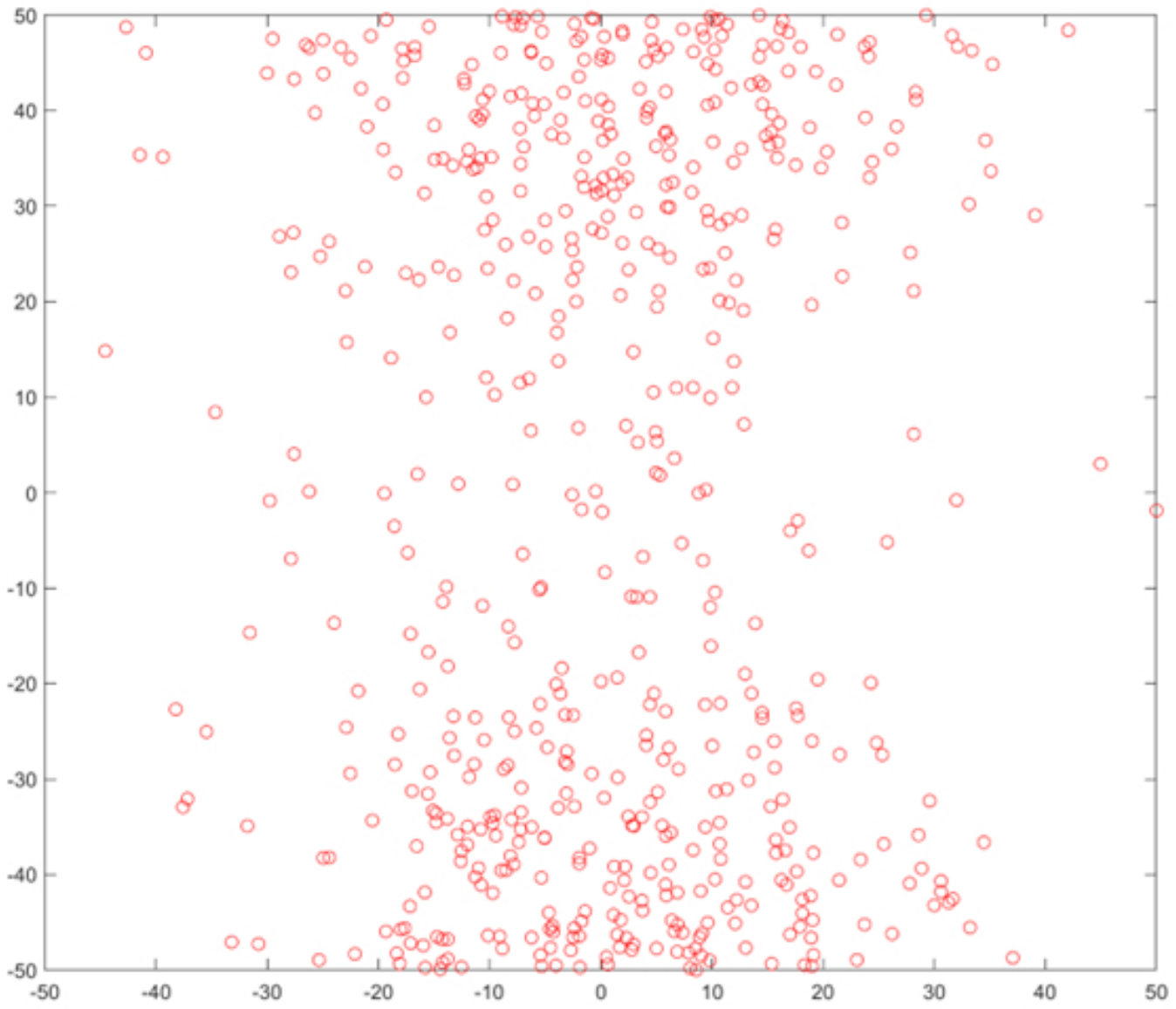}&\vspace{-1.7cm}KM: -100.1\newline CKM: -100.8\newline \textcolor{red}{KHM: -101.1}\newline \textcolor{blue}{WKHM: -100}\newline \textcolor{Ao}{KC: -98.7}&\vspace{-1.7cm}KM: -102.5\newline CKM: -102.9\newline \textcolor{red}{KHM: -103.6}\newline \textcolor{blue}{WKHM: -102.3}\newline \textcolor{Ao}{KC: -100.3}&\vspace{-1.7cm}KM: -80.3\newline \textcolor{blue}{CKM: -80.0}\newline \textcolor{Ao}{KHM: -79.8}\newline WKHM: -80.6\newline \textcolor{red}{KC: -83}&\vspace{-1.8cm}KM: 3.1\newline \textcolor{Ao}{CKM: 1.7}\newline \textcolor{blue}{KHM: 2.08}\newline WKHM: 2.4\newline \textcolor{red}{KC: 32}\\ \hline
Bi-exponential&Bi-exponential&\includegraphics[width = 0.1\linewidth]{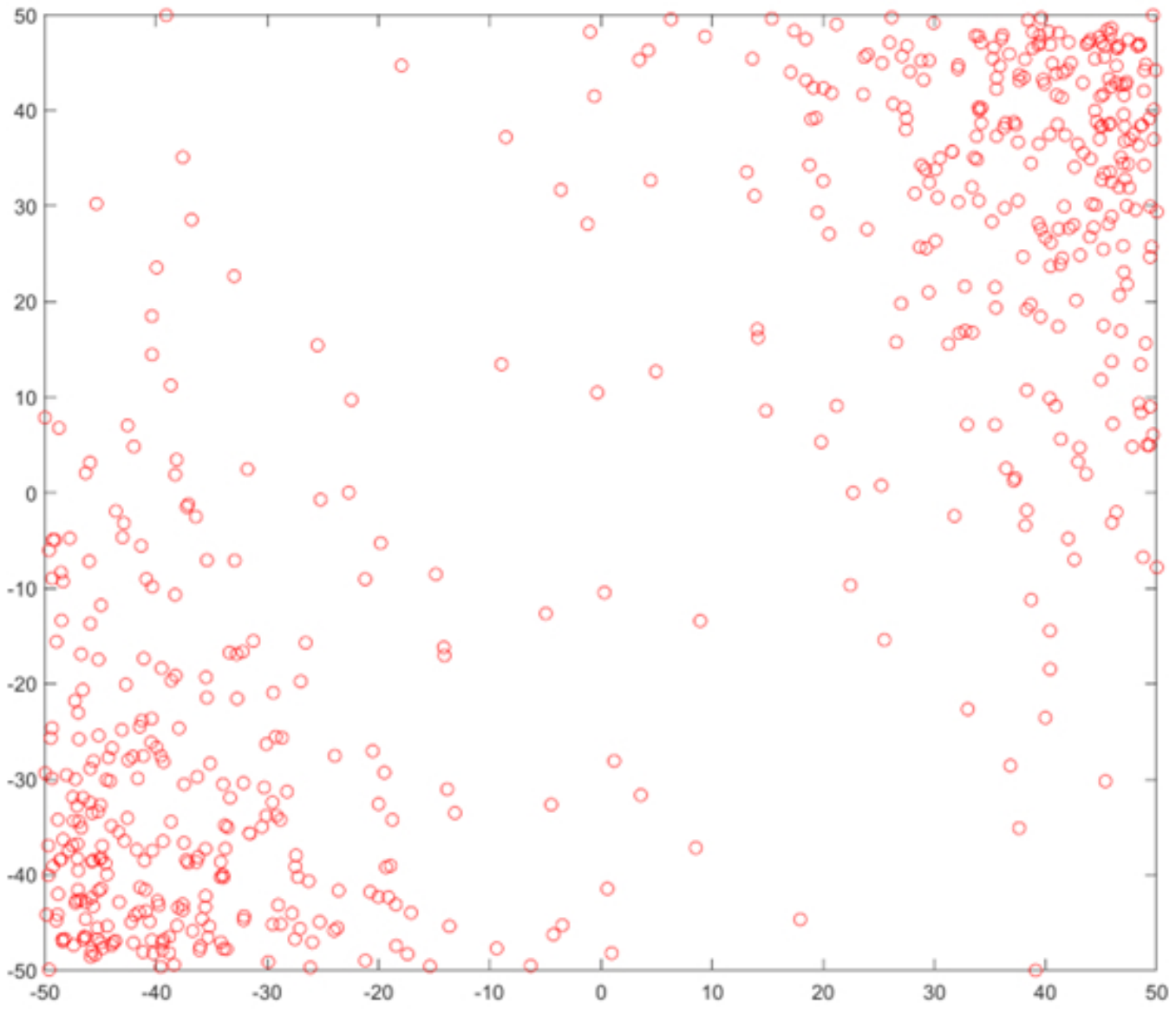}&\vspace{-1.7cm}KM: -103.62\newline \textcolor{red}{CKM: -106.0}\newline KHM: -105.5\newline \textcolor{blue}{WKHM: -103.3}\newline \textcolor{Ao}{KC: -99.4}&\vspace{-1.7cm}KM: -105.8\newline CKM: -107.1\newline \textcolor{red}{KHM: -107.7}\newline \textcolor{blue}{WKHM: -105.2}\newline \textcolor{Ao}{KC: -100.4}&\vspace{-1.7cm}KM: -80.1\newline \textcolor{Ao}{CKM: -79.3}\newline \textcolor{blue}{KHM: -79.5}\newline WKHM: -80.6 \newline \textcolor{red}{KC: -83.2}&\vspace{-1.8cm}KM: 9.5\newline \textcolor{Ao}{CKM: 1.8}\newline \textcolor{blue}{KHM: 4.6}\newline WKHM: 7.2\newline \textcolor{red}{KC: 47}\\ \hline
\end{tabular}
\caption{Monte Carlo simulations with different UE distribution skewness. The RSRP values are expressed in dBm. The worst, best, and the second best cases are represented by \textcolor{red}{red}, \textcolor{Ao}{green}, and \textcolor{blue}{blue}, respectively.}
\label{table:MC_table}
\end{table*}

\subsection{Initial RN placement:}
Once the rough number of RNs are predicted either by the MD analysis or by the ILP formulation, the next task is to optimize their locations leveraging the statistical UE location distribution.
In Table~\ref{table:MC_table} (last page) we compare the performance of the different clustering schemes in terms of the mean RSRP averaged over the network realizations, the minimum RSRP averaged over the network realization, and the worst case RSRP. Additionally, we depict the mean association ratio achieved for each algorithm, which is defined as the expected value of the ratio of the number of UEs associated to the most loaded RN with the least loaded RN. In particular, we consider four distributions of the UE locations along the $x-$ and $y-$ axes of a square deployment area.
\begin{itemize}
    \item Uniform: $f_d(x) = \frac{1}{2a}$ for $-a \leq x \leq a$.
    \item Truncated exponential: $f_d(x) = \frac{\exp(-(x + a))}{1 - \exp(-2a)}$ for $-a \leq x \leq a$.
    \item Truncated Gaussian: $\frac{f_G(x)}{F_{G}(a) - F_G(-a)}$ for $-a \leq x \leq a$.
    \item Truncated bi-exponential: 
\end{itemize}
In the above $f_G(\cdot)$ and $F_G(\cdot)$ are the \ac{PDF} and the \ac{CDF} of the standard Gaussian distribution, respectively. The combination of the $x-$coordinate (abscissa) and the $y-$coordinate (ordinate) following one of the above distributions result in ten possible UE location scenarios as shown in Table~\ref{table:MC_table}.

From a minimum RSRP perspective, the KC algorithm achieves the best guarantee across most deployment scenarios. However, it suffers from a considerably worse association ratio as compared to other algorithms. As expected, CKM performs the best in terms of the association ratio. However, in order to create equally sized cells in terms of the number of associated UEs, the CKM algorithm often creates large sized cells in terms of the geographical area. This is especially true for the case of skewed UE distributions. This results in a the cell-edge UEs of larger cells experiencing poor RSRP. Without any association constraint, WKHM achieves a good balance between the association ratio and the RSRP experience. 

{
\subsection{Controlled ablation of the WKHM weighting:}
To isolate the effect of the proposed weighting, all schemes are evaluated using the same UE realizations, number of RNs, initialization pool, and stopping tolerance. The resulting throughput and fairness statistics are summarized below.

\begin{center}
\renewcommand{\arraystretch}{1.12}
\begin{tabular}{lcccc}
\toprule
Scheme & $5$th percentile & Median & Mean & Jain index \\
& \multicolumn{3}{c}{Per-user throughput (Mbps)} & \\
\midrule
KM   & $0.6$ & $3.9$ & $4.0$ & $0.80$ \\
CKM  & $1.9$ & $4.6$ & $4.6$ & $0.88$ \\
KHM  & $1.7$ & $5.0$ & $5.0$ & $0.87$ \\
KC   & $0.4$ & $5.8$ & $5.9$ & $0.76$ \\
WKHM & $4.2$ & $6.1$ & $6.2$ & $0.96$ \\
\bottomrule
\end{tabular}
\end{center}

WKHM improves the fifth-percentile throughput by approximately $2.5$ Mbps relative to KHM and by $3.8$ Mbps relative to KC, while attaining the highest fairness index. KC produces the longest upper tail, but its low fifth-percentile throughput and Jain index reveal severe load imbalance. Consistently, the empirical SINR CDF shows zero observed outage at $\gamma=0$ dB for WKHM, whereas all comparison schemes exhibit a nonzero outage probability at the same threshold. The parameter ablation further shows that the mean SINR is stable for $\epsilon\in[10^{-7},10^{-5}]$; excessively large $p$ or $q$ makes the association too sharp and degrades SINR by increasing sensitivity to the interference geometry. The value $q=\alpha$ provides the exact received-power-fraction interpretation, while the final operating value is selected from the numerical ablation.
}

{
\subsection{Statistical evaluation for WKHM:}
All algorithmic comparisons are evaluated over $M=30$ independent spatial realizations. Within each realization, the competing methods use the same UE locations, RN cardinality, and initialization pool, so that the resulting KPI differences are paired. The reported point estimates are realization averages; uncertainty is summarized by the sample standard deviation and the $95\%$ confidence interval
\[
\overline{Y}\pm t_{0.975,M-1}\frac{s_Y}{\sqrt{M}}.
\]
Because the control-channel powers are averaged over fast fading, the principal experiment re-samples the UE process and algorithm initialization, but not instantaneous fading. Blockage and shadowing are re-sampled only in the corresponding sensitivity experiment. For pairwise throughput comparisons with WKHM, we apply a two-sided paired $t$-test to the per-realization mean throughput and use Holm correction over the four hypotheses.

The gains over KM, CKM, and KC remain significant after correction for multiple comparisons (see Tab.~\ref{tab:sig}). The WKHM-KHM confidence interval includes zero and the adjusted $p$-value is $0.054$; hence, this comparison is described as a marginal positive gain rather than a statistically significant improvement. This statistical qualification is consistent with the throughput CDF: WKHM provides the clearest advantage in the lower tail and in fairness, whereas its mean-throughput advantage over KHM is smaller.
}

\subsection{Cell load, clustering and \ac{DAS}}
In case the number of UEs associated to the RNs is below the threshold number of admissible UEs, larger cells can be created by combining multiple RNs to enable clustered cells. In this regard, first let us study the cell load in the network. In case of Voronoi cells, the cell load is challenging to derive due to the irregular shapes and sizes of the cells. However,
By using the ergodicity of the point processes, the average load can be approximated using the mean cell approach~\cite{blaszczyszyn2014user}, as: 
\begin{align}
    \rho_i &= \int_{\mathcal{A}_i} \frac{\sigma(s)}{\frac{B}{|\Phi_i|} \log_2 (1 + \xi(s))} ds \overset{(a)}{\approx} \int_{\gamma} \frac{\sigma(\gamma)}{\frac{B}{|\Phi_i|} \log_2 (1 + \xi(\gamma))}p(\gamma) d\gamma \nonumber \\
    &\overset{(b)}{\approx} \sum_{n \in \mathcal{A}_i} \frac{\lambda_A}{\frac{B}{|\Phi_i|} \log_2(1 + \xi_n)}
\end{align}
where the expectation with respect to the distance is converted in step (a) into an expectation with respect to the SINR variations, the CDF of which is precisely the standard success probability. The value of $\sigma(\gamma)$ is still challenging to derive, since this represents the traffic generated by the UEs that are experiencing an SINR of $\gamma$. To approximate this further, we calculate the cell load numerically, given the \ac{SINR} distribution within that cell for given realization of the UEs.
\begin{figure}
    \centering
    \includegraphics[width = 0.65\linewidth]{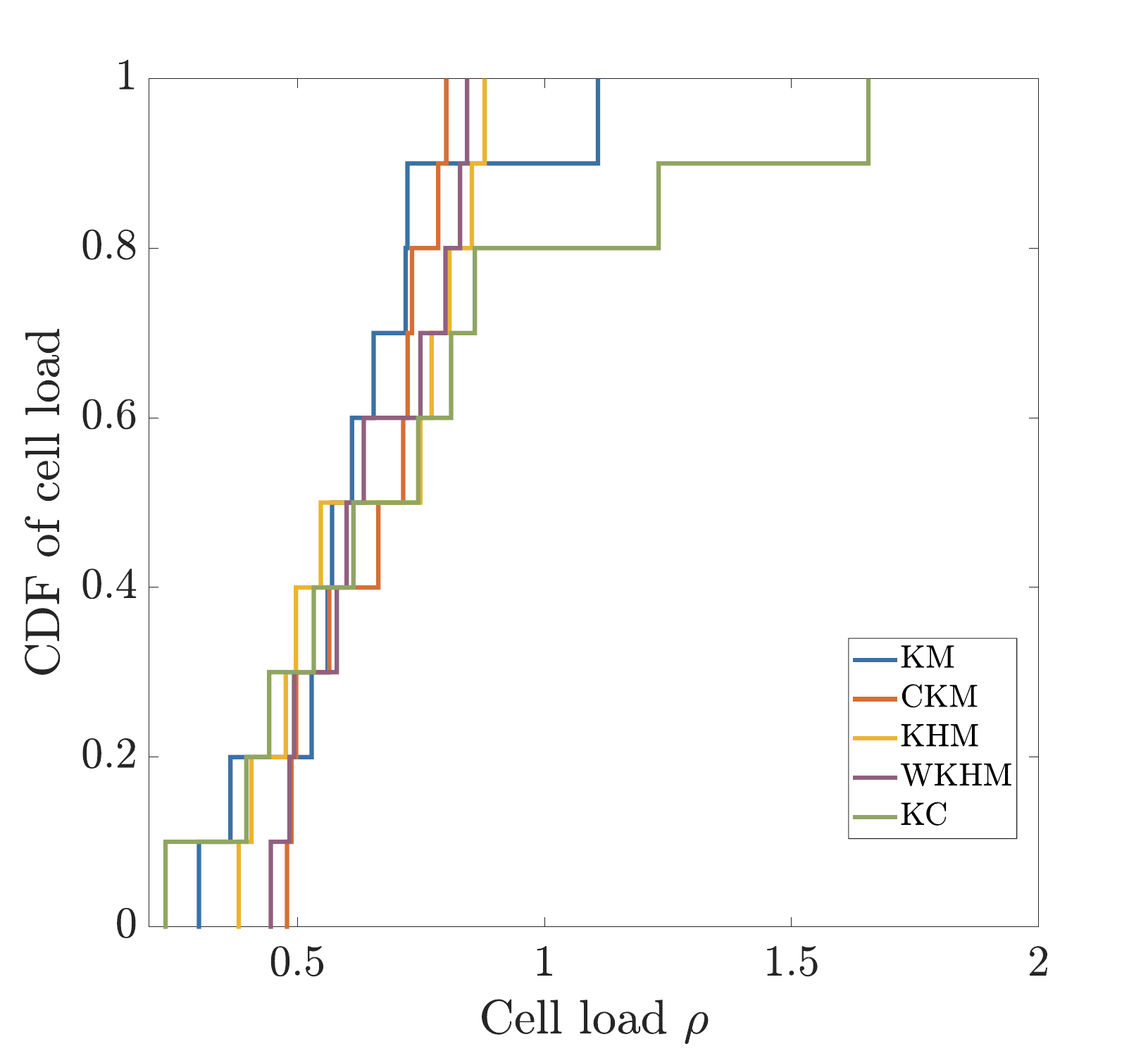}
    \caption{CDF of cell load for different initial placement schemes. Here $N_{\rm r} = 10$ and the deployment is in a 200 m $\times$ 200 m square.}
    \label{fig:loadcdf}
\end{figure}
A cell load of less than 1 corresponds to under load, i.e., the cell capacity is more than the total UE traffic generated in the cell. While, a cell load of greater than 1 refers to the overloaded case where the UEs arrive in the network at a faster rate as compared to the rate of service.

Fig.~\ref{fig:loadcdf} shows the cell-load distribution for a single realization. KM and KC each produce overloaded cells (one and two, respectively), indicating poor dimensioning, along with highly underloaded cells ($\rho \leq 0.25$). In contrast, the remaining algorithms yield only underloaded cells. This motivates cell clustering, which can merge under- and overloaded cells when load is driven by cell size rather than UE density. WKHM and CKM exhibit low load variance. Clustering gains, illustrated for WKHM initialization, increase with the maximum UE association and exceed 1.5 dB under full clustering.

\begin{figure}
\subfloat[]{\includegraphics[width = 0.45\linewidth]{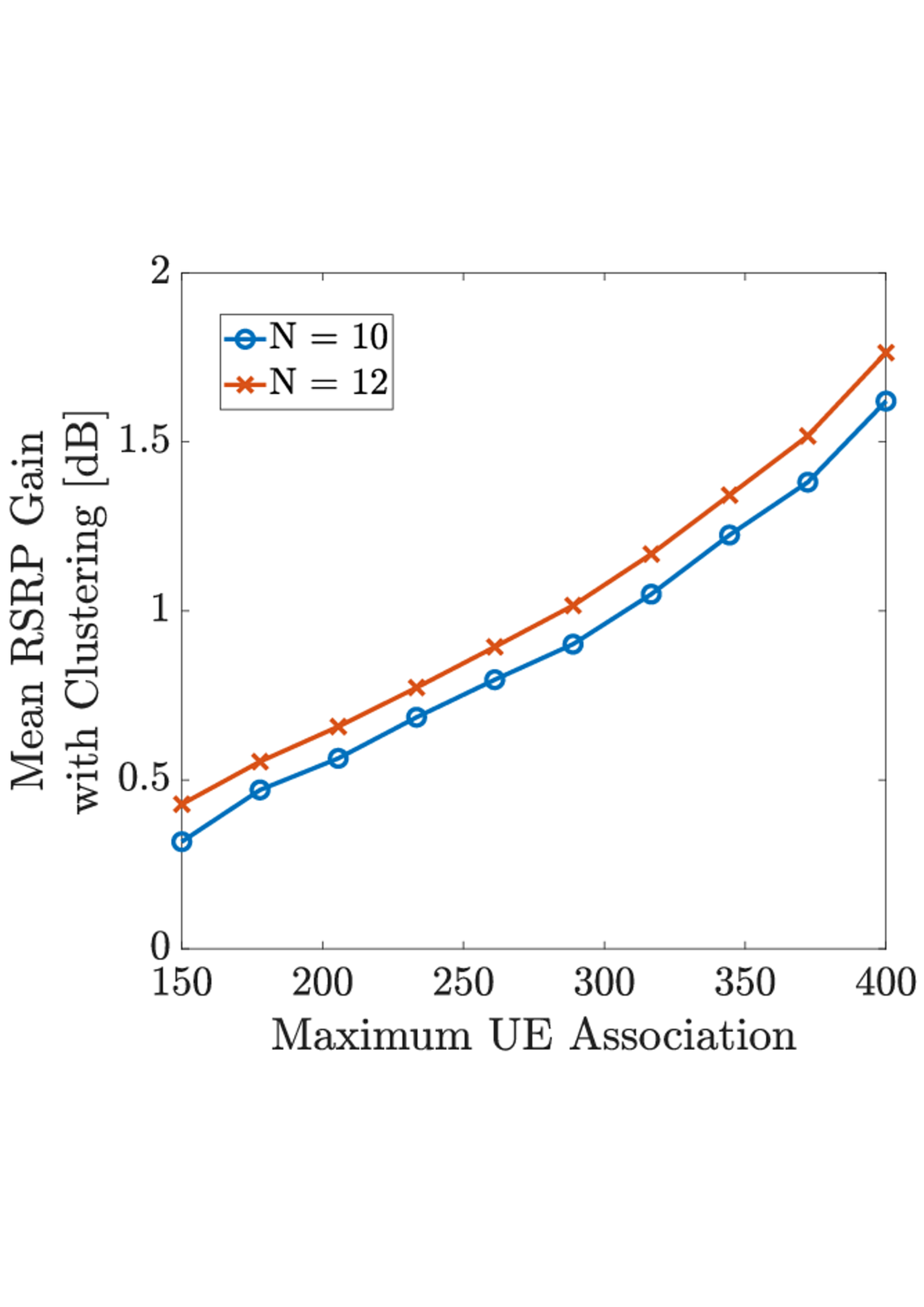} \label{fig:cluster_gain}}
\hfil
\subfloat[]{\includegraphics[width = 0.45\linewidth]{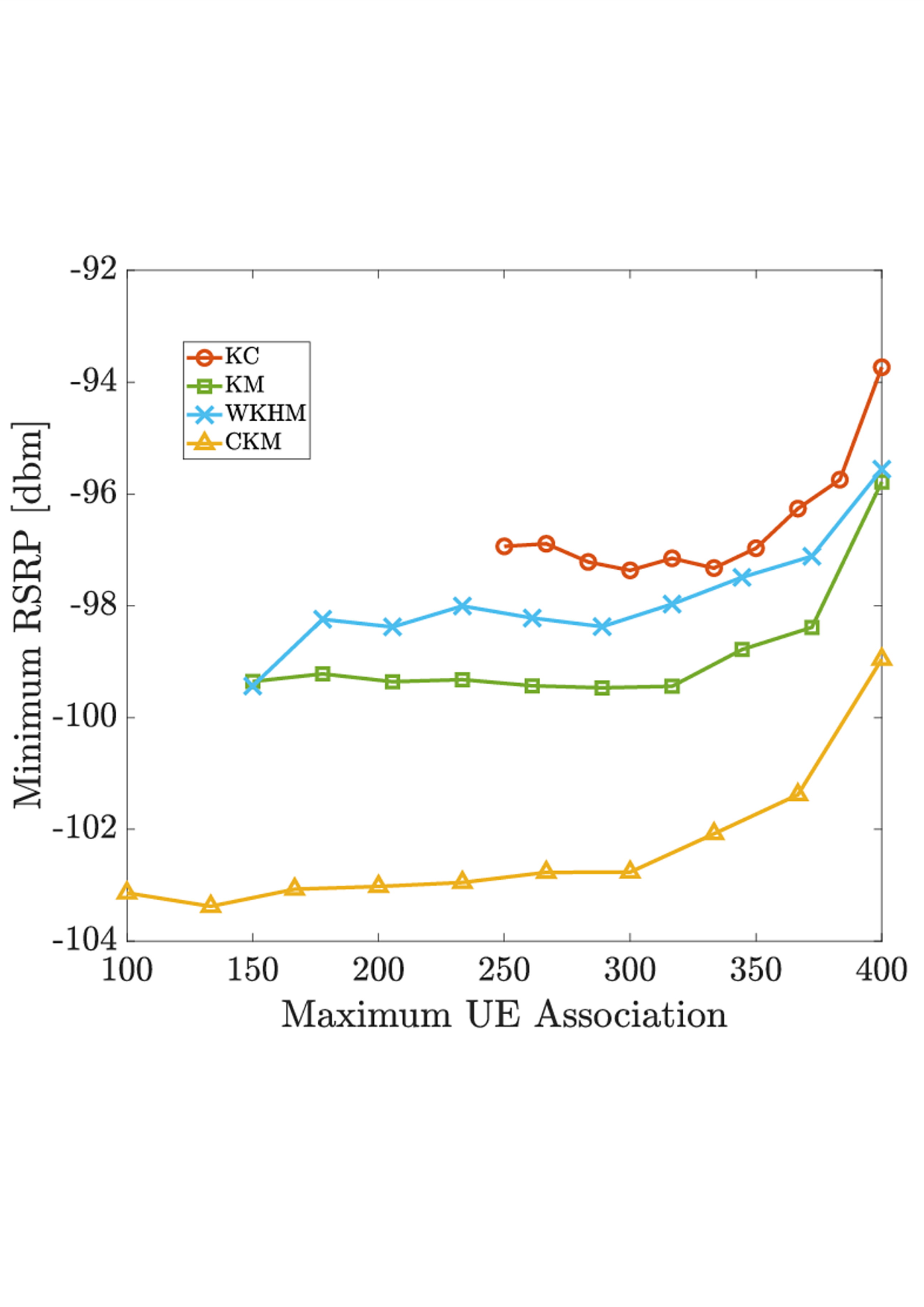}\label{fig:cluster_gain_min_rsrp}}
\caption{(a) Gain with RN clustering with respect to the number of admissible UEs per cell. Here, total number of UEs is 400. The dimensions are 500 m $\times$ 500 m and (b) Clustering gain for different initial placement strategies.}
\end{figure}


Fig.~\ref{fig:cluster_gain_min_rsrp} reveals that the impact of clustering is dependent on the initial placement scheme. Although improvement is observed whenever clustering is feasible, for a given threshold on the number of admissible UEs, all candidate algorithms may not be feasible. For example, with 400 UEs and 10  RNs, an admission threshold of 100-120 UEs results in only the CKM algorithm being feasible. All other algorithms have at least one RN with more UEs than the admission threshold. Hence, CKM must necessarily be employed for initial placement. In case the admission threshold increases to beyond 150 UEs, both KM and WKHM become feasible algorithms and with clustering, we observe that WKHM results in the best RSRP performance. Finally, when the admission threshold is increased to beyond 250, KC is included in the feasible set of the algorithms and also results in the best RSRP performance.

\begin{figure}
    \centering
    \includegraphics[width = \linewidth]{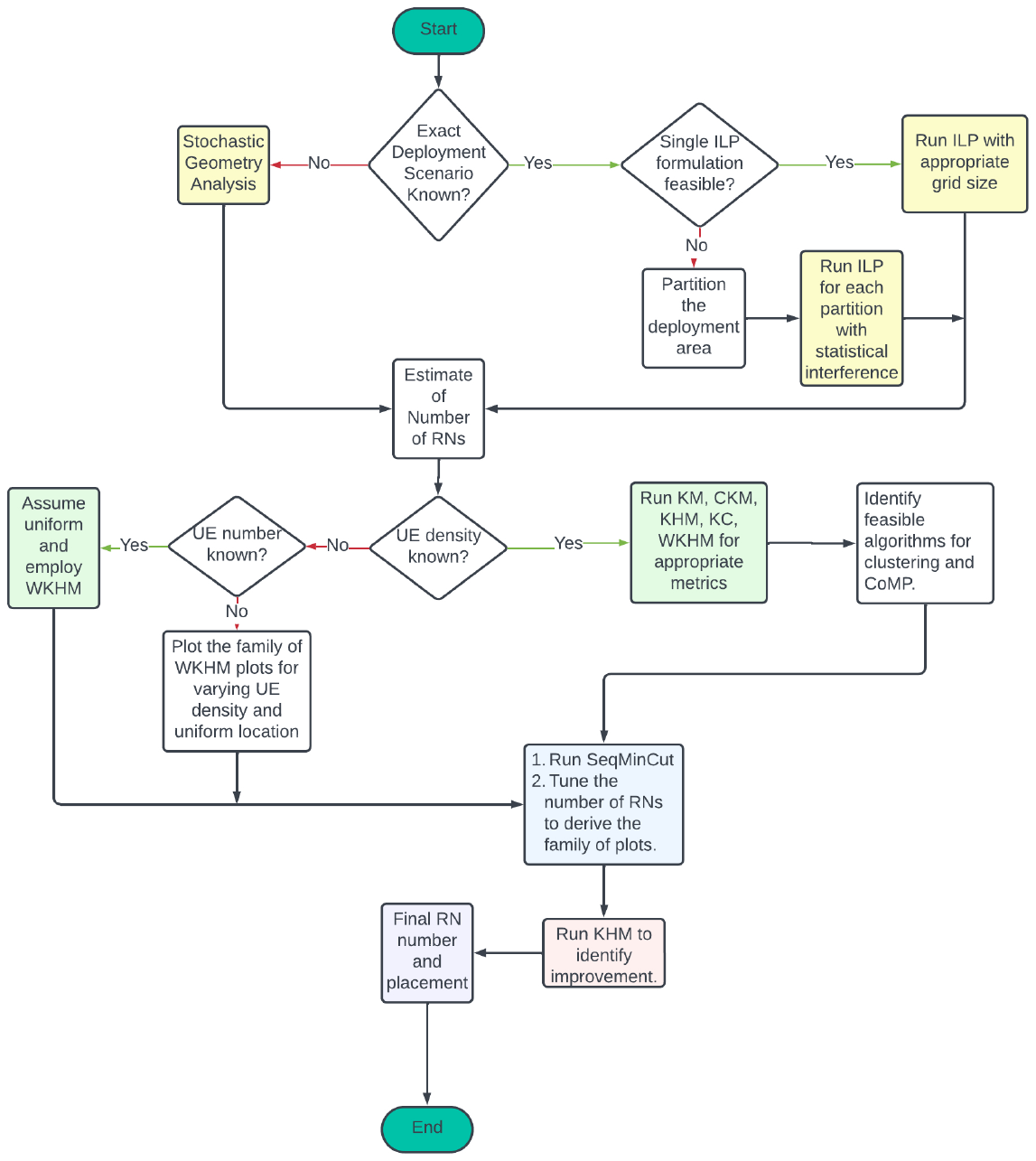}
    \caption{An example implementation of the proposed framework.}
    \label{fig:flowchart}
\end{figure}
\subsection{Framework for Network Planning and Scalability:} The above steps can be combined into a complete indoor 5G planning workflow, summarized in Fig.~\ref{fig:flowchart}. When only the deployment area is known, \ac{SG} enables initial dimensioning of the number of \acp{RN}. If the deployment map is available and computationally feasible, an \ac{ILP} framework determines RN locations and counts. With known UE density, clustering algorithms refine RN placement; otherwise, uniform UE locations reduce the problem to \ac{RSRP} optimization. Finally, \texttt{SeqMinCut} clusters RNs to form a \ac{DAS}, further improving UE performance. The full workflow is provided as a MATLAB toolbox~\cite{codes}.

Since ILP is computationally prohibitive for large areas or fine grids, we propose an iterative hybrid approach (Fig.~\ref{fig:Joint}) that combines local ILP optimization with \ac{SG}-based interference approximation. The region is partitioned into sub-regions, where ILP is applied sequentially while RNs in other areas are uniformly deployed according to the MD framework. RN locations are fixed after each step, and the process continues until all sub-regions are optimized.

\begin{figure}
    \centering
    \includegraphics[width=\linewidth]{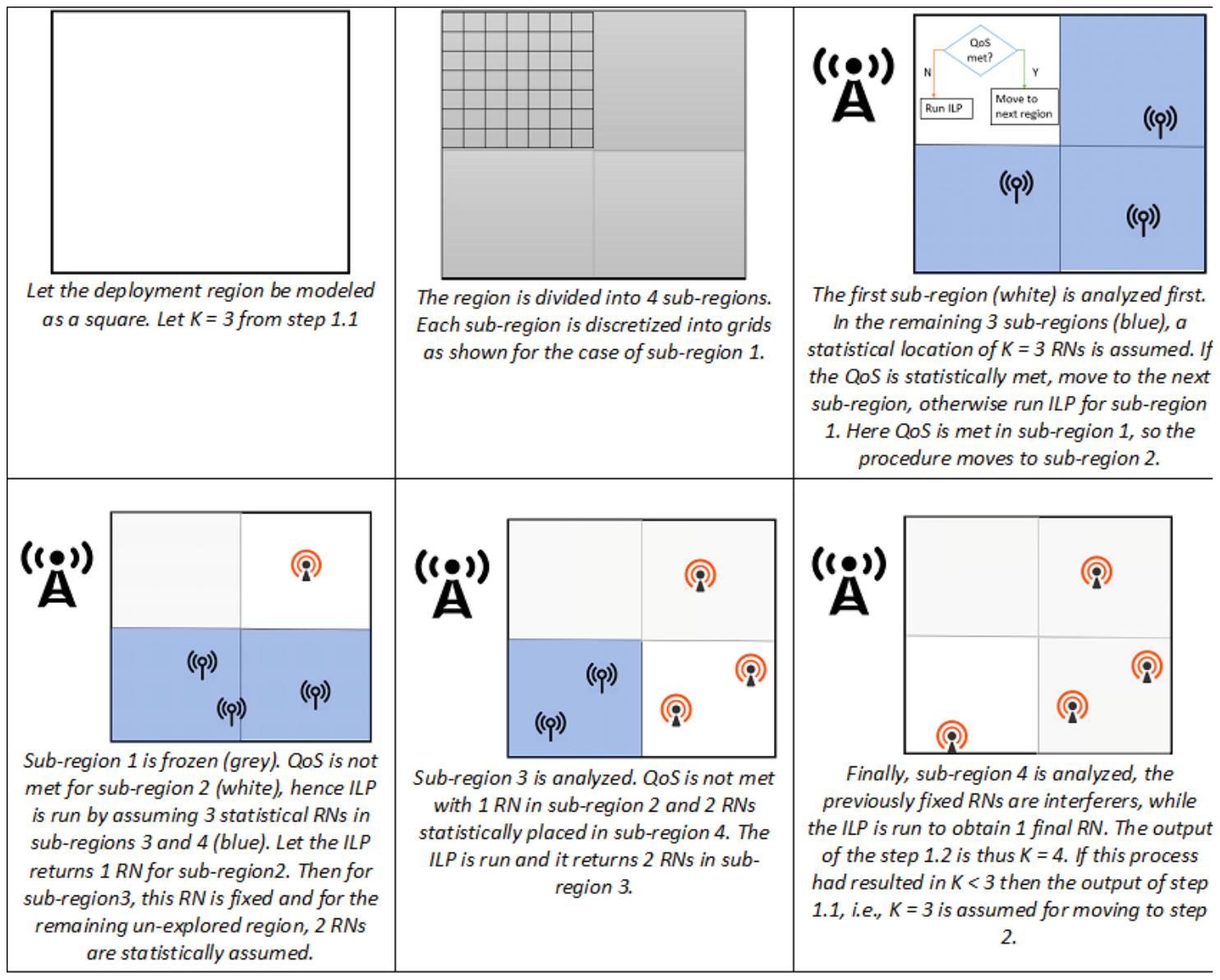}
    \caption{Iterative framework that leverages SG and ILP for real maps.}
    \label{fig:Joint}
\end{figure}

\section{Conclusions}
\label{sec:Con}
We considered the RN placement problem in indoor enterprise 5G networks and provided a comprehensive methodology for network planning. Stochastic geometry provides an efficient tool for initial network planning when the exact deployment map is not known. However, if the exact map is known, an integer linear program method is proposed that simultaneously gives the number and locations of RNs. Then, based on the distribution of the UE locations in the network, we explored a variety of clustering algorithms to obtain the location of the RNs. In particular, we proposed a novel clustering algorithm WKHM that aims to strike a balance between the UE association variation across RNs and the received RSRP. To further augment the RSRP performance of the UEs, we proposed an RN clustering algorithm to form larger cells under a UE admission constraint. We provided extensive numerical results to highlight the salient advantages of our framework. Finally, we organized the derived framework as a MATLAB toolbox which we make publicly available for network designers and researchers.

\appendices


\bibliography{references.bib}
\bibliographystyle{ieeetr}
\end{document}